\documentclass[11pt]{article}
\usepackage{paralist}
\usepackage{color}
\usepackage{bm}
\usepackage{enumitem}

\usepackage[T1]{fontenc}
\usepackage{soul}

\usepackage[cm]{fullpage}

\usepackage[algo2e,ruled,linesnumbered,vlined,procnumbered]{algorithm2e} 
\let\oldnl\nl
\newcommand{\nonl}{\renewcommand{\nl}{\let\nl\oldnl}}

\usepackage{amsfonts}
\usepackage{amssymb,mathtools}
\usepackage{multirow}
\usepackage{amsmath}
\usepackage{constants}
\usepackage{amsthm}
\usepackage[normalem]{ulem}
\makeatletter
\newtheorem*{rep@theorem}{\rep@title}
\newcommand{\newreptheorem}[2]{%
\newenvironment{rep#1}[1]{%
 \def\rep@title{#2 \ref{##1}}%
 \begin{rep@theorem}}%
 {\end{rep@theorem}}}
\makeatother

\newcount\shortyear\newcount\shorthour\newcount\shortminute
\shorthour=\time\divide\shorthour by 60\shortyear=\shorthour
\multiply\shortyear by 60\shortminute=\time\advance\shortminute by
-\shortyear \shortyear=\year\advance\shortyear by -1900

\def\zeit{\number\shorthour:\ifnum\shortminute<10 0\number\shortminute
\else\number\shortminute\fi}

\usepackage{ifpdf}
\newcommand{\mydriver}{hypertex}
\ifpdf
 \renewcommand{\mydriver}{pdftex}
\fi
\usepackage[breaklinks,\mydriver]{hyperref}

\usepackage[margin=1in]{geometry}

\newcommand{\child}{\textrm{c}}

\newcommand{\buffer}{\textrm{X}}

\newcommand{\EE}{\mathcal{E}}
\newcommand{\nil}{\textsc{nil}}

\newcommand{\poly}{\textrm{poly}}

\newcommand{\R}{\mathbb{R}}

\newcommand{\1}{\textbf{1}}

\newcommand{\vect}[1]{\ensuremath{\mathbf{#1}}}
\newcommand{\mat}[1]{\ensuremath{\mathbf{#1}}}

\theoremstyle{plain}
\newtheorem{theorem}{Theorem}[section]
\newtheorem{fact}[theorem]{Fact}
\newtheorem{lemma}[theorem]{Lemma}

\newtheorem{claim}[theorem]{Claim}
\newtheorem{invariant}[theorem]{Invariant}

\newtheorem{conjecture}[theorem]{Conjecture}

\newreptheorem{theorem}{Theorem}
\newreptheorem{lemma}{Lemma}
\newtheorem{definition}[theorem]{Definition}

\theoremstyle{definition}

\newcommand{\junk}[1]{{}}

\providecommand{\abs}[1]{\lvert#1\rvert} 
\newconstantfamily{c}{symbol=c}

\newconstantfamily{smallconst}{symbol=\kappa}

\title{Dynamic Effective Resistances and Approximate Schur Complement on Separable Graphs\footnote{The research leading to these results has received funding from the
		European Research Council under the European Union's Seventh	Framework Programme (FP/2007-2013) / ERC Grant Agreement no. 340506.}}
\author{Gramoz Goranci\footnote{University of Vienna, Faculty of Computer Science, Vienna, Austria. E-mail: \texttt{gramoz.goranci@univie.ac.at}.}
	\and
	Monika Henzinger\footnote{University of Vienna, Faculty of Computer Science, Vienna, Austria. E-mail: \texttt{monika.henzinger@univie.ac.at}.}
	\and 
	Pan Peng\footnote{Department of Computer Science, University of Sheffield, Sheffield, UK. E-mail: \texttt{p.peng@sheffield.ac.uk}. Work
		done in part while at the Faculty of Computer Science, University of Vienna, Austria.}
}
\date{}

\begin{document}


%
\begin{titlepage}
	\maketitle
	\thispagestyle{empty}
%
%
%
%



\begin{abstract}
We consider the problem of dynamically maintaining (approximate) all-pairs effective resistances in separable graphs, which are those that admit an $n^{c}$-separator theorem for some $c<1$. We give a fully dynamic algorithm that maintains  $(1+\varepsilon)$-approximations of the all-pairs effective resistances of an $n$-vertex graph $G$ undergoing edge insertions and deletions with $\tilde{O}(\sqrt{n}/\varepsilon^2)$ worst-case update time and $\tilde{O}(\sqrt{n}/\varepsilon^2)$ worst-case query time, if $G$ is guaranteed to be $\sqrt{n}$-separable (i.e., it is taken from a class satisfying a $\sqrt{n}$-separator theorem) and its separator can be computed in $\tilde{O}(n)$ time. Our algorithm is built upon a dynamic algorithm for maintaining \emph{approximate Schur complement} that approximately preserves pairwise effective resistances among a set of terminals for separable graphs, which might be of independent interest.

We complement our result by proving that for any two fixed vertices $s$ and $t$, no incremental or decremental algorithm can maintain the $s-t$ effective resistance for $\sqrt{n}$-separable graphs with worst-case update time $O(n^{1/2-\delta})$ and query time $O(n^{1-\delta})$ for any $\delta>0$, unless the Online Matrix Vector Multiplication (OMv) conjecture is false. 

We further show that for \emph{general} graphs, no incremental or decremental algorithm can maintain the $s-t$ effective resistance problem with worst-case update time $O(n^{1-\delta})$ and query-time $O(n^{2-\delta})$ for any $\delta >0$, unless the OMv conjecture is false.
	\end{abstract}
\end{titlepage}
\section{Introduction}\label{sec:intro}
Effective resistances and the closely related electrical flows  are basic concepts for resistor networks~\cite{doyle84} and were found to be very useful in the design of graph algorithms, e.g., for computing and approximating maximum flow~\cite{ChristianoKMST11,madry2013,Madry16}, random spanning tree generation~\cite{MST15fast,Sch17almost}, multicommodity flow~\cite{KMP12faster}, oblivious routing~\cite{HHNRR08}, and graph sparsification~\cite{SpielmanS11,DinitzKW15}. They also have found applications in social network analysis, e.g., for measuring the similarity of vertices in social networks~\cite{LZ18:kirchhoff}, in machine learning, e.g., for Gaussian sampling~\cite{cheng2015efficient} and in chemistry, e.g., for measuring chemical distances~\cite{klein1993resistance}. Previous research has studied the problem of how to quickly compute and approximate the effective resistances (or equivalently, \emph{energies} of electrical flows; see Appendix~\ref{app:energy} for more discussions), as such algorithms can be used as a crucial subroutine for other graph algorithms. For example, one can $(1+\varepsilon)$-approximate the $s-t$ effective resistance in  $\tilde{O}(m+n\varepsilon^{-2})$~\cite{DurfeeKPRS16} and $\tilde{O}(m \log (1/\varepsilon))$~\cite{CohenKMPPRX14} time, respectively, in any $n$-vertex $m$-edge weighted graph, for any two vertices $s,t$. (Throughout the paper, we use $\tilde{O}$ to hide polylogarithmic factors, i.e., $\tilde{O}(f(n))=O(f(n)\cdot \poly\log f(n))$.) There are also algorithms that find $(1+\varepsilon)$-approximations to the effective resistance between every pair of vertices in $\tilde{O}(n^2/\varepsilon)$ time~\cite{JS18:sketch}. In order to exactly compute the $s-t$ (or single-pair) and all-pairs effective resistance(s), the current fastest algorithms run in times $O(n^\omega)$ (by using the fastest matrix inversion algorithm~\cite{bunch1974,ibarra1982}) and $O(n^{2+\omega})$, respectively, where $\omega < 2.373$ is the matrix multiplication exponent~\cite{williams2012}. In planar graphs, the algorithms for exactly computing $s-t$ and all-pairs effective resistance(s) run in times $O(n^{\omega/2})$~(by the nested dissection method for solving linear system in planar graphs \cite{LiptonRT79}) and $O(n^{2+\omega/2})$, respectively.

A natural algorithmic question is how to efficiently maintain the effective resistances \emph{dynamically}, i.e., if the graph undergoes edge insertions and/or deletions, and the goal is to support the update operations and query for the effective resistances as quickly as possible, rather than having to recompute it from scratch each time. Besides the potential applications in the design of other (dynamic) algorithms, it is also of practical interest, e.g., to quickly report the (dis)similarity between any two nodes in a social network in which its members and their relationship are constantly changing. So far our understanding towards this question is very limited: for exact maintenance, the only approach (for single-pair effective resistance) we are aware of is to invoke the dynamic matrix inversion algorithm which gives $O(n^{1.575})$ update time and $O(n^{0.575})$ query time or $O(n^{1.495})$ update time and $O(n^{1.495})$ query time~\cite{San04:dynamic}; for $(1+\varepsilon)$-approximate maintenance, we can maintain the spectral sparsifier of size $n\poly(\log n, \varepsilon^{-1})$ with $\poly(\log n, \varepsilon^{-1})$ update time~\cite{AbrahamDKKP16}, while answering each query will cost $\Theta(n\poly(\log n, \varepsilon^{-1}))$ time. (Subsequent to the Arxiv submission~\cite{GHP18:dynaER} of this paper, Durfee et al. obtained a fully dynamic algorithm that maintains  $(1+\varepsilon)$-approximations to all-pairs effective resistances of an unweighted, undirected multi-graph with $\tilde{O}(m^{4/5}\varepsilon^{-4})$ expected amortized update and query time~\cite{DGGP18:dynER}.) 

In this paper, we study the problem of dynamically maintaining the (approximate) effective resistances in \emph{separable graphs}, which are those that satisfies an $n^c$-separator theorem for some $c<1$. Interesting classes of separable graphs include planar graphs, minor free graphs, bounded-genus graphs, almost planar graphs (e.g., road networks)~\cite{lipton1979separator}, most $3$-dimensional meshes~\cite{MTTV1997separators} and many real-world networks (e.g., phone-call graphs, Web graphs, Internet router graphs)~\cite{BBK03compact}. In the static setting, effective resistances (or electrical flows) in planar/separable graphs have been utilized by Miller and Peng~\cite{MillerP13} to obtain the first $\tilde{O}(\frac{m^{6/5}}{\varepsilon^{\Theta(1)}})$ time algorithm for approximate maximum flow in such graphs, and have also been studied by Anari and Oveis Gharan~\cite{NO15:TSP} in the analysis of an approximation algorithm for Asymmetric TSP. We now give the necessary definitions to state our results.

\paragraph{Effective Resistances.} Let $G=(V,E,\vect{w})$ be a undirected weighted graph with $\vect{w}(e)>0$ for any $e\in E$. Let $\mat{A}$ denote its weighted adjacency matrix and $\mat{D}$ denote the weighted degree diagonal matrix. Let $\mat{L}=\mat{D}-\mat{A}$ denote the \emph{Laplacian} matrix of $G$. Let $\mat{L}^\dagger$ denote the Moore-Penrose pseudo-inverse of the Laplacian of $G$. Let $\1_u \in \R^V$ denote the indicator vector of vertex $u$ such that $\1_u(v)=1$ if $v=u$ and $0$ otherwise. Let $\chi_{s,t}=\1_s-\1_t$. Given any two vertices $u,v\in V$, the \emph{$s-t$ effective resistance} is defined as $R_G(s,t):=\chi_{s,t}^T\mat{L}^\dagger\chi_{s,t}.$

\paragraph{Separable graphs.} Let $\mathcal{C}$ be a class of graphs that is closed under taking subgraphs. We say that $\mathcal{C}$ satisfies a \emph{$f(n)$-separator theorem} if there are constants $\alpha <1$ and $\beta>0$ such that every graph in $S$ with $n$ vertices has a cut set with at most $\beta f(n)$ vertices that separates the graph into components with at most $\alpha n$ vertices each~\cite{lipton1979separator}. 
In this paper we are particularly interested in the class of graphs that satisfies an $n^{1/2}$-separator theorem, which include the class of planar graphs, $K_t$-minor free graphs and bounded-genus graphs, etc., though our approach can also be generalized to other class of graphs that satisfies a $n^{c}$-separator theorem, for some $c<1$. In the following, we call a graph \emph{$f(n)$-separable} if it is a member of a class that satisfies an $f(n)$-separator theorem.

\paragraph{\ } We would like to quickly maintain the exact or a good approximation of the $s-t$ effective resistances in a $\sqrt{n}$-separable graph that undergoes edge insertions and deletions, for all pairs $s,t\in V$. We call this the \emph{dynamic all-pairs effective resistances problem.} Our goal is to solve this problem with both small update and query times. More precisely, our data structure supports the following operations. 
\begin{itemize}
	\item \textsc{Insert}$(u,v,w)$: Insert the edge $(u,v)$ of weight $w$ to $G$, provided that the updated graph remains $\sqrt{n}$-separable.
	\item \textsc{Delete}$(u,v)$: Delete the edge $(u,v)$ from $G$.   
	\item \textsc{EffectiveResistance}$(s,t)$: Return the exact or approximate value of the effective resistance between $s$ and $t$ in the current graph $G$.
\end{itemize}  

We remark that our algorithm can be extended to handle operations \textsc{Increase}($u,v,\Delta$) and \textsc{Decrease}($u,v,\Delta$) that increases and decreases the weight of any existing edge $(u,v)$ by $\Delta$, respectively, as one can simply delete the edge first and then insert it again with the corresponding new weight. For our lower bound, we will consider the \emph{incremental} (or \emph{decremental}) $s-t$ effective resistance problem, that is, $s,t$ are two vertices fixed at the beginning, and only operations \textsc{Insert} \& \text{Decrease} (or \textsc{Delete} \& \textsc{Increase}) and \textsc{EffectiveResistance} are allowed. The basic idea is that in the incremental (or decremental) setting, the effective resistances are monotonically decreasing (or increasing)~(see e.g., \cite{ChristianoKMST11}). 
For any $\varepsilon \in (0,1)$, we say that an algorithm is a $(1+\varepsilon)$-approximation to $R_G(s,t)$ if \textsc{EffectiveResistance}$(s,t)$ returns a positive number $k$ such that $(1-\varepsilon) \cdot R_G(s,t) \leq k \leq (1+\varepsilon) R_G(s,t)$. 

\subsection{Our Results}
We give a fully dynamic algorithm for maintaining $(1+\varepsilon)$-approximations of all-pairs and single-pair effective resistance(s) with small update and query times for any $\sqrt{n}$-separable graph, if its separator can be computed fast. Throughout the paper, all the running times of our algorithms are measured in \emph{worst-case} performance. All our algorithms are randomized, and the performance guarantees hold with probability at least $1-n^{-c}$ for some $c\geq 1$. Specifically, we show the following theorem. 
\begin{theorem}\label{thm:upperbound}
	Let $G$ denote a dynamic $n$-vertex graph under edge insertions and deletions. Assume that $G$ is $\sqrt{n}$-separable and its separator can be computed in $s(n)$ time, throughout the updates. There exist fully dynamic algorithms that maintain  $(1+\varepsilon)$-approximations of 
	\begin{itemize}
		\item the all-pairs effective resistances with $\tilde{O}(\frac{\sqrt{n}}{\varepsilon^2}+\frac{s(n)}{\sqrt{n}})$ update time and $\tilde{O}(\frac{\sqrt{n}}{\varepsilon^2})$ query time;
		\item the $s-t$ effective resistance with $\tilde{O}(\frac{\sqrt{n}}{\varepsilon^2}+\frac{s(n)}{\sqrt{n}})$ update time and $O(1)$ query time.
	\end{itemize} 
	In particular, if $s(n)=\tilde{O}(n)$, then our update times are $\tilde{O}(\frac{\sqrt{n}}{\varepsilon^2})$.
\end{theorem}

By using the well known facts that a balanced separator of size $O(\sqrt{n})$ for planar graphs (and bounded-genus graphs) can be computed in $O(n)$ time~\cite{lipton1979separator}, and for $K_t$-minor-free graphs (for any fixed integer $t>0$) in $O(n^{1+\delta})$ time, for any constant $\delta > 0$~\cite{kawar2010}, we obtain dynamic algorithms for the effective resistances for planar and minor-free graphs with $\tilde{O}(\sqrt{n}/\varepsilon^2)$ and $\tilde{O}(\sqrt{n}/\varepsilon^2+n^{1/2+\delta})$ update time, respectively. 

The performance of our dynamic algorithm in planar graphs almost matches the best-known dynamic algorithm for $(1+\varepsilon)$-approximate all-pairs shortest path in planar graphs with $\tilde{O}(\sqrt{n})$ update and query time~\cite{AbrahamCG12}, though our approaches are different. This is interesting as the shortest path corresponds to flows with controlled $\ell_1$ norm while the energy of electrical flows (i.e., effective resistance) corresponds to those with minimum $\ell_2$ norm.     

In order to design a dynamic algorithm for effective resistances of separable graphs (i.e., to prove Theorem~\ref{thm:upperbound}), we give a fully dynamic algorithm that efficiently maintains an \emph{approximate Schur complement}~\cite{KyngLPSS16,KyngS16,DurfeeKPRS16} of such graphs (see Section~\ref{sec:dynamic_ASC}), which might be of independent interest.
Approximate Schur complement can be treated as a \emph{vertex sparsifier} that preserves pairwise effective resistances among a set of terminals (see Section~\ref{sec:asc_property}). Therefore, our algorithm is a dynamic algorithm for \emph{vertex effective resistance sparsifiers} with sublinear (in $n$) update time for separable graphs. The problem of dynamically maintaining graph \emph{edge sparsifiers} has received attention very recently. For example, Abraham et al. presented fully dynamic algorithms that maintain cut and spectral sparsifiers with poly-logarithmic update times~\cite{AbrahamDKKP16}. Formally, we prove the following theorem. 

\begin{theorem}\label{thm:dynamic_asc}
	For an $n$-vertex $\sqrt{n}$-separable graph $G$ whose separator can be computed in $s(n)$ time, and a terminal set $K \subseteq V$ with $|K|\leq O(\sqrt{n})$, there exists a fully dynamic algorithm that maintains a $(1+\delta)$-approximate Schur complement with respect to $K'$ such that $K \subseteq K'$ and $|K'|=O(\sqrt{n})$, while achieving $\tilde{O}(\sqrt{n}/\delta^{2}+\frac{s(n)}{\sqrt{n}})$ update time. 
	Furthermore, our algorithm supports terminal additions as long as $|K|\leq O(\sqrt{n})$.
\end{theorem}

We complement our algorithm by giving a conditional lower bound for any \emph{incremental} or \emph{decremental} algorithm that maintains \emph{single-pair} effective resistance of a $\sqrt{n}$-separable graph. Our lower bound is established from the \emph{Online Matrix Vector Multiplication} (\emph{OMv}) \emph{conjecture} (see Section \ref{sec:conjecture}).
\begin{theorem}\label{thm:lowerbound_separable}
	No incremental or decremental algorithm can maintain the (exact) $s-t$ effective resistance in $\sqrt{n}$-separable graphs on $n$ vertices with both $O(n^{\frac12-\delta})$ worst-case update time and $O(n^{1-\delta})$ worst-case query time for any $\delta >0$, unless the OMv conjecture is false.
\end{theorem}

We note that there are very few conditional lower bounds for dynamic \emph{planar/separable} graphs, as most known reductions are highly non-planar. The only recent result that we are aware of is by Abboud and Dahlgaard~\cite{AbboudD16}, who showed that under some popular conjecture, no algorithm for dynamic shortest paths or maximum weight bipartite matching in planar graphs has both updates and queries in amortized $O(n^{1/2-
	\delta})$ time, for any $\delta>0$.

We also give a stronger conditional lower bound for the same problem in \emph{general} graphs, which shows that it is hard to maintain effective resistances with both sublinear (in $n$) update and query times for general graphs, even for the incremental or decremental setting. 
\begin{theorem}\label{thm:lowerbound_general}
	No incremental or decremental algorithm can maintain the (exact) $s-t$ effective resistance in general graphs on $n$ vertices with both $O(n^{1-\delta})$ worst-case update time and $O(n^{2-\delta})$ worst-case query time for any $\delta >0$, unless the OMv conjecture is false.
\end{theorem}

We remark that both lower bounds for separable and general graphs hold for any algorithm with sufficiently high accurate approximation ratio, and both lower bounds for incremental algorithms hold even if only edge insertions are allowed (see Section~\ref{sec: lowerBound}). 

\paragraph*{Comparison to~\cite{GHP17:sparsifier}}
In our previous work~\cite{GHP17:sparsifier}, we gave a fully dynamic algorithm for $(1+\varepsilon)$-approximating all-pairs effective resistances for planar graphs with $\tilde{O}(r/\varepsilon^2)$ update time and $\tilde{O}((r+n/\sqrt{r})/\varepsilon^2)$ query time, for any $r$ larger than some constant. The algorithm can also be generalized to $\sqrt{n}$-separable graphs, and we also provided a conditioned lower bound for any approximation algorithm of the $s-t$ effective resistance in general graphs in the \emph{vertex-update} model. However, besides the apparent improvement of the performance of the dynamic algorithm (i.e., we reduce the best trade off between update time and query time from $\tilde{O}(n^{2/3})$ and $\tilde{O}(n^{2/3})$ to $\tilde{O}(n^{1/2})$ and $\tilde{O}(n^{1/2})$), our current work also improves over and differs from~\cite{GHP17:sparsifier} in the following perspectives.
\begin{itemize}
\item Our algorithm dynamically maintains the approximate Schur complement of a separable graphs by maintaining a separator tree of such graphs, rather than their \emph{$r$-divisions} as used in~\cite{GHP17:sparsifier}. In fact, we do not believe purely $r$-divisions based algorithms will achieve the performance as guaranteed by our new algorithm. This is evidenced by previous dynamic algorithms for maintaining reachability in directed planar graphs by Subramanian \cite{Subramanian93}, $(1+\varepsilon)$-approximating to all-pairs shortest paths by~Klein and Subramanian~\cite{KleinS98}, exactly maintaining $s-t$ max-flow in planar graphs by Italiano et al.~\cite{ItalianoNSW11}, all of which are based on $r$-divisions and have running times of order $n^{2/3}$ (and some of which have been improved by using other approaches). 
\item Our current lower bound is much stronger than the previous one: the previous lower bound only holds for general graphs and the \emph{vertex-update} model, where nodes, not edges, are turned on or off, and its proof was based on a simple relation between $s-t$ connectivity and $s-t$ effective resistance $R_G(s,t)$ (i.e., if $s,t$ is connected iff $R_G(s,t)$ is not infinity). In contrast, our new lower bounds hold for separable graphs (and also general graphs) and the edge-update model. The corresponding proofs exploit new reductions from the OMv problem to the $5$-length cycle detection and triangle detection problems in separable graphs and general graphs, respectively, which might be of independent interest, and the latter problems are related to the effective resistances (see Section~\ref{sec:proof_lower_separable}).    
\end{itemize}


\subsection{Our Techniques}
Our dynamic algorithm for maintaining an Approximate Schur complement (ASC) w.r.t. a set of terminals for separable graphs is built upon maintaining a \emph{separator tree} of such graphs and two properties (called \emph{transitivity} and \emph{composability}) of ASCs. Such a tree can be constructed very efficiently by recursively partitioning the subgraphs using separators. Slightly more formally, each node in the tree corresponds to a subgraph of the original graph and contains a subset of vertices as its boundary vertices which in turn are treated as terminals. For each node $H$, we will maintain an ASC $H'$ of $H$ w.r.t its terminals. We will guarantee throughout all the updates that the ASC of any node can be computed efficiently in a bottom-up fashion, by the above two properties of ASCs. This stems from the fact that we only need to recompute the ASCs of nodes that lie on a path from a \emph{constant} number leaves to the node of interest. Since each such path has length $O(\log n)$ and the recomputation of ASC of one node takes time $\tilde{O}(\sqrt{n})$, the update time will be guaranteed to be $\tilde{O}(\sqrt{n})$. For the detailed implementation, we need to overcome the difficulty that the error in the approximation ratio might accumulate through this recursive computation and an update might require to change the set of boundary vertices of many nodes, thus resulting in a prohibitive running time. We remark that though the idea of using separator tree of planar/separable graphs is standard~(e.g.,~\cite{EppsteinGIS96}), the main novelty of our algorithm is to use such a tree as the backbone to dynamically maintain the approximate Schur complement.

To obtain our dynamic algorithms for all-pairs effective resistance, we appropriately declare and add new terminals whenever we get a new query, and then run the above dynamic algorithm for ASC with respect to the corresponding terminal set.

To obtain our lower bound, we provide new reductions from the Online Boolean Matrix-Vector Multiplication (OMv) problem to the incremental or decremental single-source effective resistance problem. More specifically, given an OMv instance with vectors $\vect{u,v}$ and a matrix $\mat{M}$, we construct a $\sqrt{n}$-separable graph $G$ such that $\vect{u}\mat{M}\vect{v}=1$ if and only if there exists a cycle of length $5$ incident to some vertex $t$ in $G$. This $5$-length cycle detection problem in turn can be solved by inspecting the diagonal entry corresponding to $t$ of the inverse of a matrix that is defined from $G$. Furthermore, the diagonal entry of this matrix is inherently related to the effective resistance~\cite{MNSUW17:spectrum}. By appropriately dynamizing the graph $G$ and using the time bounds for the OMv problem from the conjecture, we get the conditional lower bound for separable graphs. For general graphs, the lower bound is proved in a similar way, except that the constructed graph is different and we instead use a relation between effective resistance and triangle detection problem. That is, we first reduce the OMv problem to the $t$-triangle detection problem such that the OMv instance satisfies $\vect{u}\mat{M}\vect{v}=1$ if and only if there exists a triangle incident to some vertex $t$ in the constructed $G$. The latter problem can again be solved by checking the diagonal entry corresponding to $t$ of some matrix, which in turn encodes the effective resistance of between $t$ and a properly specified vertex $s$.

\paragraph*{Other Related Work.}
Previous work on dynamic algorithms for planar or plane graphs include: shortest paths~\cite{KleinS98,AbrahamCG12,ItalianoNSW11}, $s-t$ min-cuts/max-flows~\cite{ItalianoNSW11}, reachability in directed graphs~\cite{Subramanian93,IKLS17decremental,DS07dynamic}, ($k$-edge) connected components~\cite{EppsteinGIS96,HIKLRS2017contracting}, the best swap and the minimum spanning forest~\cite{EppsteinGIS96}. There also exist work on dynamic algorithms for $\sqrt{n}$-separable graphs, e.g., on transitive closure and $(1+\varepsilon)$-approximation of all-pairs shortest paths~\cite{karczmarz18}.

As mentioned before, subsequent to our Arxiv submission, Durfee et al.~\cite{DGGP18:dynER} obtained a dynamic all-pairs effective resistances algorithm with $\tilde{O}(m^{4/5}\varepsilon^{-4})$ expected amortized update and query time, against an oblivious adversary. This algorithm uses ideas stemmed from this paper, in particular, one of their key ideas is to dynamically maintain an approximate Schur complement. If restricted to separable graphs, the running times of their algorithm are worse than ours. It is also interesting to note that for the (simpler) offline dynamic effective resistance problems, i.e., the sequence of updates and queries are given as an input, Li et al.~\cite{LPYZ18} recently gave an incremental algorithm with $O(\frac{\poly\log n}{\varepsilon^2})$ amortized update and query time for general graphs.

\section{Preliminaries} \label{sec: prelim}


\subsection{Properties of Separable Graphs} \label{sec: sepTree}
\paragraph*{Separator Trees.} Let $G=(V,E)$ be a sparse, $O(\sqrt{n})$-separable graph. For an edge-induced subgraph $H$ of $G$, any vertex that is incident to vertices not in $H$ is called a \emph{boundary vertex}. We let $\partial(H)$ denote the set of \emph{boundary vertices} belonging to $H$. All other vertices incident to edges from only $H$ will be called \emph{interior vertices} of $H$.

A hierarchical decomposition of $G$ is obtained by recursively partitioning the graph using separators into edge-disjoint subgraphs (called regions), where the removal of each separator partitions the subgraph into two two edge-disjoint subgraph. This decomposition is represented by a binary (decomposition) tree $\mathcal{T}(G)$, which we refer to as a \emph{separator tree} of $G$. For any subgraph $H$ of $G$, we use $H \in \mathcal{T}(G)$ to denote that $H$ is a node of $\mathcal{T}(G)$ (to avoid confusion with the vertices of $G$, we refer to the vertices of $\mathcal{T}(G)$ as nodes). The \emph{height} $\eta(H)$ of a node is the number of edges in the longest path between that node and a leaf. In addition, let $S(H)$ denote a balanced separator of the subgraph $H$. Formally, $\mathcal{T}(G)$ satisfies the following properties:
\begin{enumerate}[noitemsep]
\item The root node of $\mathcal{T}(G)$ is the graph $G$.
\item A non-leaf node $H \in \mathcal{T}(G)$ has exactly two children $\child_1(H)$, $\child_2(H)$ and a balanced separator $S(H)$ such that $\child_1(H) \cup \child_2(H) = H$, $V(\child_1(H)) \cap V(\child_2(H)) = S(H)$ and $E(\child_1(H)) \cap E(\child_2(H)) = \emptyset$.
\item For a node $H \in \mathcal{T}(G)$, the set of boundary vertices $\partial(H) \subseteq V(H)$ is defined recursively as follows:
\begin{itemize}[noitemsep]
\item If $H$ is the root of $\mathcal{T}(G)$, i.e., $H=G$, then $\partial(G) = S(G)$. 
\item Otherwise, $\partial(H) = S(H) \cup (\partial(P) \cap V(H))$, where $P$ is the parent of $H$ in $\mathcal{T}(G)$.
\end{itemize}
\item For each node $H \in \mathcal{T}(G)$ and its children $\child_1(H)$, $\child_2(H)$, we have $\partial(\child_1(H)) \cup \partial(\child_2(H)) \supseteq \partial(H)$.
\item The number of boundary vertices per node $H \in \mathcal{T}(G)$, i.e., $|\partial(H)|$, is bounded by $O(\sqrt{n})$.
\item There are $O(\sqrt{n})$ leaf subgraphs in $\mathcal{T}(G)$, each having at most $O(\sqrt{n})$ edges.
\item The height of the tree $\mathcal{T}(G)$ is $O(\log n)$, i.e., $\eta(G) = O(\log n)$.
\item Each edge $e \in E$ is contained in a unique leaf subgraph of $\mathcal{T}(G)$.
\end{enumerate}

The lemma below shows that a separator tree can be constructed with an additional $\log n$ factor overhead in the running time for computing a separator. For the sake of completeness we include its proof in Appendix~\ref{app: prelim}.

\begin{lemma} \label{lem: sepTreeTime}
Let $G=(V,E)$ be a $O(\sqrt{n})$-separable graph whose balanced separator can be computed in $s(n)$ time. There is an algorithm that computes a separator tree $\mathcal{T}(G)$ in $O(s(n) \log n)$ time. 
\end{lemma}

\subsection{The Online Boolean Matrix-Vector Multiplication (OMv) Conjecture}\label{sec:conjecture}
Our lower bound will be built upon the following OMv problem and conjecture.
\begin{definition}
	In the \emph{Online Boolean Matrix-Vector Multiplication (OMv)} problem, we are given an integer $n$ and an $n\times n$ Boolean matrix $\mat{M}$. Then at each step $i$ for $1\leq i\leq n$, we are given an $n$-dimensional column vector $\vect{v}_i$, and we should compute $\mat{M}\vect{v}_i$ and output the resulting vector before we proceed to the next round.
\end{definition}
\begin{conjecture}[OMv conjecture~\cite{henzinger2015unifying}]~\label{conjecture}
	For any constant $\varepsilon>0$, there is no $O(n^{3-\varepsilon})$-time algorithm that solves OMv with error probability at most $1/3$. 
\end{conjecture}
We will work on a related problem which is called the $\vect{u}\mat{M}\vect{v}$ problem.
\begin{definition}\label{def:umv_problem}
	In the $\vect{u}\mat{M}\vect{v}$ problem with parameters $n_1,n_2$, we are given a matrix $M$ of size $n_1\times n_2$ which can be preprocessed. After preprocessing, a vector pair $\vect{u},\vect{v}$ is presented, and our goal is to compute $\vect{u}^T\mat{M}\vect{v}$.
\end{definition}
\begin{theorem}[\cite{henzinger2015unifying}]\label{thm:umv_hard}
	Unless the OMv conjecture~\ref{conjecture} is false, there is no algorithm for the $\vect{u}\mat{M}\vect{v}$ problem with parameters $n_1,n_2$ using polynomial preprocessing time and computation time $O(n_1^{1-\delta}n_2+n_1n_2^{1-\delta})$ that has an error probability at most $1/3$, for some constant $\delta$.
\end{theorem}

\subsection{Spectral and Resistance Sparsifiers}
Below we present two notion of edge sparsifiers. The first requires that the quadratic form of the original and sparsified graph are close. The second requires that all-pairs effective resistances of the corresponding graphs are close. 
\begin{definition}[Spectral Sparsifier] \label{def: specSpar} Given a graph $G=(V,E,\vect{w})$ and $\varepsilon \in (0,1)$, we say that a subgraph $H=(V,E_H,\vect{w}_H)$ is an $(1 \pm \varepsilon)$-\emph{spectral sparsifier} of $G$ if 
	\[ \forall \vect{x} \in \mathbb{R}^{n},~(1-\varepsilon)\vect{x}^{T}\mat{L}(G)\vect{x} \leq \vect{x}^{T}{\mat{L}(H)}\vect{x} \leq (1+\varepsilon)\vect{x}^{T}\mat{L}(G)\vect{x}. \]
\end{definition}
\begin{definition}[Resistance Sparsifier] Given a graph $G=(V,E,\vect{w})$ and $\varepsilon \in (0,1)$, we say that a subgraph $H=(V,E_H,\vect{w}_H)$ is an  $(1 \pm \varepsilon)$-\emph{resistance sparsifier} of $G$ if
	\[
	\forall u,v \in V,~(1-\varepsilon)R_G(u,v) \leq R_H(u,v) \leq (1+\varepsilon) R_G(u,v).
	\]
\end{definition}
The following lemma shows that Definition~\ref{def: specSpar} is equivalent to approximating the pseudoinverse Laplacians. For the sake of completeness we include its proof in Appendix~\ref{app: prelim}.

\begin{lemma} \label{lem: equivPseudo}
Assume $G$ is connected. Then the following statements are equivalent:
\begin{enumerate}
\item  $\forall \vect{x} \in \mathbb{R}^{n},~(1-\varepsilon)\vect{x}^{T}\mat{L}(G)\vect{x} \leq \vect{x}^{T}{\mat{L}(H)}\vect{x} \leq (1+\varepsilon)\vect{x}^{T}\mat{L}(G)\vect{x}.$
\item $\displaystyle \forall \vect{x} \in \mathbb{R}^{n},~\frac{1}{(1+\varepsilon)}\vect{x}^{T}\mat{L}(G)^\dagger\vect{x} \leq \vect{x}^{T}\mat{L}(H)^\dagger\vect{x} \leq \frac{1}{(1-\varepsilon)}\vect{x}^{T}\mat{L}(G)^\dagger\vect{x}.$
\end{enumerate}
\end{lemma}

In our algorithm we use the following observations: (1) Since, by definition, the effective resistance between any two nodes $u$ and $v$ is the quadratic form defined by the pseudo-inverse of the Laplacian computed at the vector $\1_s-\1_t$, i.e., $R_G(u,v) = (\1_s-\1_t)^T\mat{L}^\dagger(\1_s-\1_t)$, it follows that the effective resistances between any two nodes in $G$ and $H$ are the same up to a $1/(1 \pm \varepsilon)$ factor. By definitions for resistance and spectral sparsifiers, and Lemma~\ref{lem: equivPseudo} we have the following fact.

\begin{fact}
	Let $\varepsilon\in (0,1)$ and let $G$ be a graph. Then every $(1\pm\varepsilon)$-spectral sparsifier of $G$ is an $1/(1\pm \varepsilon)$-resistance sparsifier of $G$. 
\end{fact}

(2) The lemma below suggests that given a graph, by decomposing the graph into several  pieces and computing a good sparsifier for each piece, one can obtain a good sparsifier for the original graph which is the union of the sparsifiers for all pieces. 

\begin{lemma}[\cite{AbrahamDKKP16}, Lemma~4.18] \label{lemm: decomposability} Let $G=(V,E,\vect{w})$ be a weighted graph whose set of edges is partitioned into $E_1,\ldots,E_\ell$. Let $H_i$ be a $(1\pm \varepsilon)$-spectral sparsifier of $G_i=(V,E_i)$, where $i=1,\ldots,\ell$. Then $H=\bigcup_{i=1}^{\ell} H_i$ is a $(1\pm \varepsilon)$-spectral sparsifier of $G$. 
\end{lemma}

\subsection{Schur Complement and Approximate Schur Complement}
For a given connected graph $G=(V,E)$ and a set $K\subset V$ of terminals with $1\leq |K|\leq |V|-1$, let $N = V \setminus K$ be the set of non-terminal vertices in $G$. The partition of $V$ into $N$ and $K$ naturally induces the following partition of the Laplacian $\mat{L}(G)$ into blocks:
\[ \mat{L}(G) = 
\begin{bmatrix}
\mat{L}_{N} & \mat{L}_{M} \\
\mat{L}^{T}_{M} & \mat{L}_K \\
\end{bmatrix}
\]
We remark that since $G$ is connected and $N$ and $K$ are non-empty, one can show that $\mat{L}_N$ is invertible. We have the following definition of Schur complement.
\begin{definition}[Schur Complement] \label{def: Schur} The (unique) \emph{Schur complement} of a graph Laplacian $\mat{L}(G)$ with respect to a terminal set $K$ is
	\[
	\mat{S}(G,K) := \mat{L}_{K} - \mat{L}^{T}_{M}\mat{L}^{-1}_{N}\mat{L}_{M}.
	\]
\end{definition} 
It is well known that the matrix $\mat{S}(G,K)$ is a Laplacian matrix for some graph $G'$. 


\begin{definition}[Approximate Schur Complement (ASC)] \label{def: approxSchurComplement} Given a graph $G=(V,E,\mat{w})$, $K \subset V$ and its Schur complement $\mat{S}(G,K)$, we say that a graph $H=(K,E_H,\mat{w}_H)$ is a $(1\pm \varepsilon)$\emph{-approximate Schur complement (abbr. $(1\pm \varepsilon)$-ASC)} of $G$ with respect to $K$ if 
	\[
	\forall \vect{x} \in \mathbb{R}^{k},~(1-\varepsilon)\vect{x}^{T}\mat{S}(G,K)\vect{x} \leq \vect{x}^{T}{\mat{L}(H)}\vect{x} \leq (1+\varepsilon)\vect{x}^{T}\mat{S}(G,K)\vect{x}. 
	\]
	Moreover, we say that $H$ is an $1$\emph{-ASC} of $G$ with respect to $K$ if $\mat{L}(H) = \mat{S}(G,K)$.
\end{definition}

Note that $(1\pm \varepsilon)$-ASC is a spectral sparsifier of Schur complement. 
Furthermore, approximate Schur complement can be computed efficiently as guaranteed in the following lemma~\cite{DurfeeKPRS16}. 
\begin{lemma} \label{lem: ApproxSchur}
	Fix $\varepsilon \in (0,1/2)$ and $\gamma \in (0,1)$, and let $G=(V,E,\vect{w})$ be a graph with $K \subset V$ and  $|K|=k$. There is an algorithm \textsc{ApproxSchur}$(G,K, \varepsilon, \delta)$ that computes a $(1 \pm \varepsilon)$-\emph{ASC} $H$ of $G$ with respect to $K$ such that the following statements hold  probability at least $1-\gamma$: 
	\begin{enumerate}
		\item The graph $H$ has $O(k\varepsilon^{-2}\log(n/\gamma))$ edges.
		\item The total running time for computing $H$ is $\tilde{O}( m \log^{3}(n / \gamma) + n \varepsilon^{-2} \log^{4} (n/\gamma))$.
	\end{enumerate}	
\end{lemma}

\section{Useful Properties of Approximate Schur Complement}
\label{sec:asc_vs}
In this section we show that Approximate Schur complement can be treated as a vertex effective resistance sparsifier, which is a small graph that (approximately) preserves the pairwise effective resistances among terminal vertices of the original graph. Then we show two important properties called \emph{transitivity} and \emph{composability} properties of ASCs, which will be exploited in our dynamic algorithms for ASCs and effective resistances.


\subsection{ASC as Vertex Resistance Sparsifier}\label{sec:asc_property}
To maintain all-pairs effective resistances efficiently, it will be useful to consider the following notion of \emph{vertex sparsifier} that preserves pairwise effective resistances among a set of terminals. 
\begin{definition}[Vertex Resistance Sparsifier (VRS)] Given a graph $G=(V,E,\vect{w})$ with $K \subset V$, we say that a graph $H=(K,E_H,\vect{w}_H)$ is an $(1 \pm \varepsilon)$-\emph{vertex resistance sparsifier}~\emph{(abbr. $(1\pm \varepsilon)$-VRS)} of $G$ with respect to $K$ if
	\[
	\forall s,t \in K,~(1-\varepsilon) R_G(s,t) \leq R_H(s,t) \leq (1+\varepsilon) R_H(s,t).
	\]
\end{definition} 

We show that ASC can be treated as a vertex resitance sparsifier. For this, we recall the following lemma which shows that the quadratic form of the pseudo-inverse of the Laplacian $\mat{L}$ is preserved by taking the quadratic form of the pseudo-inverse of its Schur complement, for demand vectors supported on the terminals. 

\begin{lemma}[\cite{MillerP13}, Lemma~5.1] \label{lemm: SchurComp} Let $\vect{d}$ be a demand vector of a graph $G$ whose vertices are partitioned into terminals $K$, and non-terminals $N$ such that only terminals have non-zero entries in $\vect{d}$. Let $\vect{d}_K$ be the restriction of $\vect{d}$ on the terminals and let $\mat{S}(G,K)$ be the Schur complement of $\mat{L}(G)$ with respect to $K$. Then 
	\[
	\vect{d}^{T} \mat{L}(G)^{\dagger} \vect{d} =  \vect{d}_K^{T} \mat{S}(G,K)^{\dagger} \vect{d}_K.
	\]
\end{lemma}

Using interchangeability between graphs and their Laplacians, we can interpret the above result in terms of graphs as well. The lemma below relates ASCs and vertex resistance sparsifiers. For the sake of completeness, we include its proof in Appendix~\ref{app: approxSchur}.
\begin{lemma}~\label{lemma:exact_schur} Let $G=(V,E,\mat{w})$ be a graph with $K \subset V$. If $H$ is an $(1\pm \varepsilon)$-\emph{ASC} of $G$ with respect to $K$, then $H$ is an \emph{$1/(1 \pm \varepsilon)$-VRS} of $G$ with respect to $K$.
\end{lemma}


\subsection{Transitivity and Composability of ASCs}
In the following, we show a \emph{transitivity} property of ASCs and then show how the ASCs of two neighboring nodes of the separator tree $\mathcal{T}(G)$ can be combined to give the ASC of their parent (called \emph{composability}), which will enable us to compute the ASCs of all nodes of $\mathcal{T}(G)$ in a bottom-up fashion. 

\paragraph{Transitivity of ASCs.} To show the transitivity property the ASCs, we will use the following lemma which establishes the connection between the Schur complement and the Laplacian of the original graph. 

\begin{lemma}[\cite{MillerP13}, Lemma B.2] \label{lem: Schur}
Let $\mat{L}(G)$ be the Laplacian of $G$ and let $\mat{S}(G,K)$ be its Schur complement. For any $x \in \mathbb{R}^{k}$ the following holds
\[
	\mat{x}^{T} \mat{S}(G,K) \mat{x} = \min_{\mat{y}}
	 \begin{bmatrix*}
           \mat{y} \\
           \mat{x} \\
         \end{bmatrix*}^T \mat{L}(G) \begin{bmatrix*}
           \mat{y} \\
           \mat{x} \\
         \end{bmatrix*}.
\]
\end{lemma}


We are now ready to show the following transitive property of ASCs.

\begin{lemma}[Transitivity of ASCs] \label{lem: multiComb}

If $H'$ is an $(1 \pm \varepsilon)$-\emph{ASC} of $G$ with respect to $K'$, and $H$ is an $(1 \pm \varepsilon)$-\emph{ASC} of $H'$ with respect to $K$, where $K' \supseteq K$, then $H$ is an $(1 \pm \varepsilon)^2$-\emph{ASC} of $G$ with respect to $K$.
\end{lemma}
\begin{proof}
Let $k=|K|$ and $k'=|K'|$. By the assumption of the lemma, the following inequalities hold:
\[
 \forall \vect{x} \in \mathbb{R}^{k'},~(1-\varepsilon)\vect{x}^{T}\mat{S}(G,K')\vect{x} \leq \vect{x}^{T}{\mat{L}(H')}\vect{x} \leq (1+\varepsilon)\vect{x}^{T}\mat{S}(G,K')\vect{x},
 \]
and
\[
 \forall \vect{x} \in \mathbb{R}^{k},~(1-\varepsilon)\vect{x}^{T}\mat{S}(H',K)\vect{x} \leq \vect{x}^{T}{\mat{L}(H)}\vect{x} \leq (1+\varepsilon)\vect{x}^{T}\mat{S}(H',K)\vect{x}.
 \]
We need to show that
\[
	\forall \vect{x} \in \mathbb{R}^{k},~ (1-\varepsilon)^2 \vect{x}^{T}{\mat{S}(G,K)}\vect{x} \leq \vect{x}^{T}{\mat{L}(H)}\vect{x} \leq (1+\varepsilon)^2 \vect{x}^{T}{\mat{S}(G,K)}\vect{x}.
\]

We first show the upper bound on $\vect{x}^{T}{\mat{L}(H)}\vect{x}$. Note that since $K' \supseteq K$, using Gaussian elimination, $\mat{S}(G,K)$ can be constructed by first constructing $\mat{S}(G,K')$ from $G$ and then constructing $\mat{S}(G,K)$ from $\mat{S}(G,K')$ using Gaussian elimination. Thus $\mat{S}(G,K)$ is the Schur complement of $\mat{S}(G,K')$ with respect to $K$. For any $\vect{x} \in \mathbb{R}^{k}$, let $\vect{y}$ be the vector that attains the minimum value in Lemma~\ref{lem: Schur} for $\mat{S}(G,K')$. If we define $\vect{x'} = \begin{bmatrix} \vect{y} &  \vect{x} \end{bmatrix}^T \in \mathbb{R}^{k'}$, then we get
\begin{align*}
 \vect{x}^{T}{\mat{L}(H)}\vect{x} & \leq (1+\varepsilon)\vect{x}^{T}\mat{S} (H',K)\vect{x} \\ 
 & \leq (1+\varepsilon) \vect{x'}^T \mat{L}(H') \vect{x'} \\ 
 & \leq (1+\varepsilon)^2 \vect{x'}^{T}\mat{S}(G,K')\vect{x'} \\
 & = (1+\varepsilon)^2 \vect{x}^{T}{\mat{S}(G,K)}\vect{x}.
\end{align*}
We now give the lower bound on $\vect{x}^{T}{\mat{L}(H)}\vect{x}$. Recall that $\mat{S}(H',K)$ is the Schur complement of $\mat{L}(H')$ with respect to $K$. For any vertex $\vect{x} \in \mathbb{R}^{k}$, let $\vect{y}$ be the vector given by Lemma~\ref{lem: Schur} for $\mat{L}(H')$. If we define $\vect{x''} = \begin{bmatrix} \vect{y} &  \vect{x} \end{bmatrix}^T \in \mathbb{R}^{k'}$, then we get
\begin{align*}
 \vect{x}^{T}{\mat{L}(H)}\vect{x} & \geq (1-\varepsilon)\vect{x}^{T}\mat{S} (H',K)\vect{x} \\ 
 & = (1-\varepsilon) \vect{x''}^T \mat{L}(H') \vect{x''} \\ 
 & \geq (1-\varepsilon)^2 \vect{x''}^{T}\mat{S}(G,K')\vect{x''} \\
 & \geq (1-\varepsilon)^2 \vect{x}^{T}{\mat{S}(G,K)}\vect{x}.
\end{align*}\end{proof}

\paragraph{Composability of ASCs.} To show the composability of ASCs, we first review an equivalent way of defining Schur complements. The main idea is to view $\mat{S}(G,K)$ as a multi-graph where each multi-edge corresponds to a walk in $G$ that starts and ends at $K$, but has all intermediate vertices in $V \setminus K$. We call such a walk a \emph{terminal-free} walk that starts and ends in $K$. Formally, a terminal-free walk \[u_0, \ldots, u_\ell\] of length $\ell$, with $u_0,u_\ell \in K$ and $u_i \in V \setminus K$, for $i=1,\ldots,\ell$ corresponds to a multi-edge between $u_0$ and $u_\ell$ in $\mat{S}(G,K)$ with weight given by
\begin{equation} \label{eq: weightSchur}
	w^{\mat{S}(G,K)}_{u_0,\ldots,u_{\ell}} = \frac{\prod_{0 \leq i \leq \ell} w^{G}_{u_i u_i+1}}{\prod_{0 < i < \ell} d^{G}_{u_i}}.
\end{equation}

This connection is formally proven in the lemma below.

\begin{lemma}[\cite{DurfeePPR17:arxiv}, Lemma 5.4] \label{lem: EquivSchur} 
Given a graph $G$ and a partition of its vertices into $K$ and $V \setminus K$, the graph $G^{K}$ obtained by forming an union over all multi-edges corresponding to terminal-free walks that start and end in $K$, with weights given by Equation (\ref{eq: weightSchur}) is exactly $\mat{S}(G,K)$.
\end{lemma}

We next show that if a graph can be viewed as a combination of two graphs along some subset of shared terminals, combining the respective sparsifiers of these two graphs in the same way gives a sparsifer for the original graph.

Formally, Let $G_1=(V_1,E_1)$ and $G_2=(V_2,E_2)$ be edge-disjoint graphs with terminals $K_1$ and $K_2$, respectively. Furthermore, assume that all vertices in the intersection of $V_1$ and $V_2$, if exist, are terminals in both graphs. That is, $(V_1 \cap V_2) \subset K_i$, for $i=\{1,2\}$. The \emph{merge} of $G_1$ and $G_2$ is the graph $G =(V_1 \cup V_2, E_1 \cup E_2)$ with terminals $K_1 \cup K_2$ formed by identifying the terminals in $S$. We denote this operation by $G := G_1 \oplus G_2$.

\begin{lemma}[Composability of Schur complement]\label{lemm: exactMerge}

Let $G := G_1 \oplus G_2$. If $H_1$ is an $1$\emph{-ASC} of $G_1$ with respect to $K_1$, and $H_2$ is an $1$\emph{-ASC} of $G_2$ with respect to $K_2$, then $H := H_1 \oplus H_2$ is an $1$\emph{-ASC} of $G$ with respect to $K$.
\end{lemma}
\begin{proof}
Note that $H_i = \mat{S}(G_i, K_i)$, for $i=\{1,2\}$, and recall that the $G_1$ and $G_2$ share the terminals in some non-empty subset $S$, i.e., $S \subset K_i$, for $i=\{1,2\}$. To prove the lemma, we need to show that 

\[ \mat{S}(G_1, K_1) \oplus \mat{S}(G_2, K_2) = \mat{S}(G,K). \] 

We do so by making use of Lemma~\ref{lem: EquivSchur}. More specifically, we argue that every multi-edge (along with its corresponding weight) in $\mat{S}(G,K)$ is contained either in $\mat{S}(G_1, K_1)$ or $\mat{S}(G_2, K_2)$. We distinguish the following cases.

(1) For any two terminals $t$ and $t'$ in $K_1 \setminus S$, we have that $\mat{S}(G_1, K_1)$ contains all the multi-edges between $t$ and $t'$ in $\mat{S}(G,K)$. This is because $G_1$ and $G_2$ are edge-disjoint, and there is no terminal-free walk between $t$ and $t'$ in $G$ that does not use a terminal in $S$. The same reasoning can be applied to terminal pairs in $K_2 \setminus S$. 

(2) For any two terminals $s$ and $t$ in $S \times K$, we have that the corresponding multi-edges in $\mat{S}(G,K)$, are either contained in $\mat{S}(G_1, K_1)$ or $\mat{S}(G_2, K_2)$. If $t \in K_1 \setminus S$ or $t \in K_2 \setminus S$, then the same reasoning as in case (1) applies. However, if $t \in S$, then $S(G_1,K_1)$ contains all the multi-edges that correspond to terminal-free walks between $s$ and $t$ that use the edges in $G_1$, and $S(G_2,K_2)$ contains all the multi-edges that correspond to terminal-free walks between $s$ and $t$ that use the edges in $G_2$.

(3) For any two terminals $t$ and $t'$ in $(K_1 \setminus S) \times (K_2 \setminus S)$, there is no terminal-free walk between $t$ and $t'$ in $G$ that does not use a terminal in $S$, since $S$ is a separator of $G$. Thus there are no multi-edges between $t$ and $t'$ in $\mat{S}(G,K)$, so the merge $\mat{S}(G_1, K_1) \oplus \mat{S}(G_2, K_2)$ correctly does not add such edges.   
\end{proof}

\begin{lemma}[Composition of ASCs] \label{lem: approxMerge}

Let $G := G_1 \oplus G_2$. If $H_1'$ is an $(1\pm \varepsilon)$\emph{-ASC} of $G_1$ with respect to $K_1$, and $H_2'$ is an $(1 \pm \varepsilon)$\emph{-ASC} of $G_2$ with respect to $K_2$, then $H' := H_1' \oplus H_2'$ is an $(1 \pm \varepsilon)$\emph{-ASC} of $G$ with respect to $K$.
\end{lemma}
\begin{proof}
First, let $H_1$ be an $1$-ASC of $G_1$ with respect to $K_1$, and $H_2$ be an $1$-ASC of $G_2$ with respect to $K_2$. By Lemma~\ref{lemm: exactMerge}, $H := H_1 \oplus H_2$ is an $1$-ASC of $G$ with respect to $K$, i.e., $\mat{L}(H) = \mat{S}(G,K)$. Now note that we can treat $H_i$ and $H_i'$, for $i=\{1,2\}$ as graphs defined on the same vertex set $V(H)$, by adding appropriate isolated vertices. By assumption, each $H_i'$ is an $(1\pm \varepsilon)$-spectral sparsifier of $H_i$ and thus, applying the Decomposition Lemma~\ref{lemm: decomposability} gives that $H':= H_1' \oplus H_2'$ is an $(1\pm \varepsilon)$-spectral sparsifier of $H$, or equivalently, $H'$ is an $(1 \pm \varepsilon)$-ASC of $G$.
\end{proof}

\section{Dynamic Algorithms for Effective Resistances in Separable Graphs}
In this section, we first present our fully dynamic algorithm for maintaining a $(1\pm \delta)$-approximate Schur complement (i.e., prove Theorem~\ref{thm:dynamic_asc}) and then use it give a dynamic algorithm for $(1+\varepsilon)$-approximating all-pairs effective resistances in separable graphs and prove Theorem~\ref{thm:upperbound}. For simplicity, we assume that the separator of $G$ can be computed in $\tilde{O}(n)$ time. 
\subsection{Dynamic Approximate Schur Complement}\label{sec:dynamic_ASC}

Let $\delta \in (0,1)$. Let $K\subset V$ be a set of terminals with $|K|\leq O(\sqrt{n})$. We give a data-structure for maintaining a $(1 \pm \delta)$-ASC of a $O(\sqrt{n})$-separable graph $G$ with respect to a set $K'$ of $O(\sqrt{n})$ vertices (which contains the terminal set $K$) that supports \textsc{Insert} and \textsc{Delete} operations as defined before. In addition, it supports the following operation:
\begin{itemize}

\item \textsc{AddTerminal}$(u)$: Add the vertex $u$ to the terminal set $K$, as long as $\abs{K} \leq O(\sqrt{n})$. 
\end{itemize}
%


\paragraph*{Data Structure.} Throughout we compute and maintain a balanced separator $S(G)$ of $G$ that contains $K$ and satisfies that $|S(G)|\leq O(\sqrt{n})$. We let $K'=S(G)$ and we will maintain a $(1\pm \delta)$-ASC of $G$ w.r.t. $K'$. By definition of boundary vertices,  $K'=\partial(G)$. Let $\delta'=\frac{\delta}{c\log n+1}$ for some constant $c$. In our dynamic algorithm, we will maintain a separator tree $\mathcal{T}(G)$ (see {Section~\ref{sec: sepTree}) such that for each node $H \in \mathcal{T}(G)$, we maintain its separator $S(H)$ and a set $\buffer(H)$ of edges of $H$, which is initially empty, and an ASC $H'$ of $H$ w.r.t. $\partial(H)$. 
Throughout the updates, the set $\buffer(H)$ will denote the subset of edges which are only contained in $H$ while contained in neither of its children. Let $\mathcal{D}(G,\delta)$ denote such a data-structure. We recompute $\mathcal{D}(G,\delta)$ every $\Theta(\sqrt{n})$ operations using the initialization below.


\paragraph*{Initialization.} We show how to efficiently compute the ASC $H'$ for each node $H$ from $\mathcal{T}(G)$. We do this in a bottom-up fashion by first calling 
Algorithm~\ref{alg: approxSchurNode} on each leaf node and then on the non-leaf nodes, where $\textsc{ApproxSchur}$ is the procedure from Lemma~\ref{lem: ApproxSchur}.

In what follows, whenever we compute an approximate Schur complement, we assume that procedure \textsc{ApproxSchur} from Lemma~\ref{lem: ApproxSchur} is invoked on the corresponding subgraph and its boundary vertices, with error $\delta'$ and error probability $\gamma = 1/n^3$. In the following, we will assume that all the calls to the \textsc{ApproxSchur} are correct.




\begin{algorithm2e}[htb!]
\caption{\textsc{ApproxSchurNode}$(H, \partial(H),\delta')$}
\label{alg: approxSchurNode}

Set $\gamma = 1/n^{3}$. \\
\If{$H$ is a leaf} 
{
	Set $H' \gets \textsc{ApproxSchur}(H,\partial(H),\delta', \gamma)$. 
} 

\If{$H$ is a non-leaf}
{
    Let $\child_1(H),\child_2(H)$ be the children of $H$. \\
    Let $\child_i(H)'$ be the ASC of $\child_i(H)$, for $i=1,2$.  \\ 
    Set $R \gets \child_1(H)' \oplus_{\phi} \child_2(H)'$ and $E(R) \gets E(R) \cup \buffer(H)$.  \\
    Set $H' \gets \textsc{ApproxSchur}(R,\partial(H),\delta', \gamma)$.
}

\Return $H'$. 

\end{algorithm2e}

%
%

The following lemma shows that after invoking Algorithm~\ref{alg: approxSchurNode} in a bottom-up fashion, we have computed the ASC for every node in $\mathcal{T}(G)$.
\begin{lemma} \label{lem: dsCorrectness}

Let $H \in \mathcal{T}(G)$ be a node of height $\eta(H) \geq 0$ and $\emph{X}(H) = \emptyset$. Then $H' = \textsc{ApproxSchur}$ $\textsc{Node}(H,\partial(H),\varepsilon)$ is an $(1\pm \delta')^{\eta(H)+1}$\emph{-ASC} of $H$ with respect to $\partial(H)$. 
\end{lemma}
\begin{proof}
We proceed by induction on $\eta(H)$. For the base case, i.e., $\eta(H) = 0$, $H$ is a leaf node. By Lemma~\ref{lem: ApproxSchur} and Algorithm~\ref{alg: approxSchurNode}, $H'$ is indeed a $(1 \pm \delta')$-ASC of $H$ with respect to $\partial(H)$. 
	
	Let $H$ be a non-leaf node, i.e. $\eta(H) > 0$. Let $\child_1(H), \child_2(H)$ and $\child_1'(H), \child_2'(H)$ be defined as in Algorithm~\ref{alg: approxSchurNode}. By properties (2), (3) and (4) of $\mathcal{T}(G)$ and the fact that $X(H)=\emptyset$, we have $H = \child_1(H) \oplus \child_2(H)$. 
	By induction hypothesis, it follows that $\child_i(H)'$ is an $(1\pm \delta')^{\eta(\child_i(H)) + 1}$-ASC of $\child_i(H)$, for $i=1,2$. Using Lemma~\ref{lem: approxMerge} and since $\eta(\child_i(H)) + 1 = \eta(H)$, for $i=1,2$, we get that $R := \child_1(H)' \oplus \child_2(H)'$ is an $(1\pm \delta')^{\eta(H)}$-ASC of $H$ with respect to $V(R) := \partial(\child_1(H)) \cup \partial(\child_2(H))$. Now, since $V(R) \supseteq \partial(H)$ by property (4) of $\mathcal{T}(G)$ and by Lemma~\ref{lem: ApproxSchur}, it follows that $H'$ is an $(1 \pm \delta')$-ASC of $R$ with respect to $\partial(H)$. Finally, applying Lemma~\ref{lem: multiComb} on $R$ and $H'$ we get that $H'$ is an $(1 \pm \delta')^{\eta(H) + 1}$-ASC of $H$.
\end{proof}

Next we analyze the running time of the initialization and recomputation procedure. The lemma below shows that the ASC of any node in $\mathcal{T}(G)$ can be computed in $\tilde{O}(\sqrt{n}/\delta^{2})$.

\begin{lemma} \label{lemm: computeAscNode}
	Let $H \in \mathcal{T}(G)$ and assume that $\abs{\emph{\buffer}(H)} \leq O(\sqrt{n})$. Then we can compute an \emph{ASC} $H'= \textsc{ApproxSchurNode}(H,\partial(H),\varepsilon)$ of $H$ in $\tilde{O}(\sqrt{n}/\delta^{2})$ time.
\end{lemma}
\begin{proof}
	We distinguish two cases. First, if $H$ is a leaf node, then by property (5) of $\mathcal{T}(G)$, we have that $\abs{E(H)} \leq O(\sqrt{n})$. The latter along with Lemma~\ref{lem: ApproxSchur} (2) imply the time to compute $H'$ is $\tilde{O}(\sqrt{n}/\delta'^{2})$. Second, if $H$ is a non-leaf node, then by Lemma~\ref{lem: ApproxSchur} (1) we know that $\abs{E(\child_i(H)')} \leq \tilde{O}(\sqrt{n}/\delta'^{2})$, for $i = 1,2$. Since by assumption $\abs{\buffer(H)} \leq O(\sqrt{n})$, we get that $\abs{R \cup \buffer(H)} \leq \tilde{O}(\sqrt{n}/\delta'^{2})$. Thus, the time to compute $H'$ on top of $R \cup \buffer(H)$ is bounded by $\tilde{O}(\sqrt{n}/\delta'^{2}) = \tilde{O}(\sqrt{n}/\delta^{2})$ (again by Lemma~\ref{lem: ApproxSchur} (2) and the choice of $\delta'$).
\end{proof}

We now analyze the running time for initializing our data-structure. Let $T_{\mathcal{D}(G)}$ denote the time required to compute $\mathcal{D}(G)$. 

\begin{lemma} \label{lem: DSrunningTime}
	
	The time $T_{\mathcal{D}(G)}$ required to compute $\mathcal{D}(G)$ is $\tilde{O}(n / \delta^{2})$.
\end{lemma}
\begin{proof}
	By Lemma~\ref{lem: sepTreeTime} recall that we can construct $\mathcal{T}(G)$ in $\tilde{O}(n)$ time. Note that by construction of the separator tree, the number of non-leaf nodes is bounded by the number of leaf nodes. Since there there are $O(\sqrt{n})$ leaf nodes, the total number of nodes in $\mathcal{T}(G)$ is $O(\sqrt{n})$. By Lemma~\ref{lemm: computeAscNode} we get that the time needed to compute an ASC $H'$ for every node $H \in \mathcal{T}(G)$ is $\tilde{O}(\sqrt{n}/\delta^{2})$. Combining the above bounds gives that $T_{\mathcal{D}(G)}$ is $\tilde{O}(n / \delta^{2})$.
\end{proof}
Since $\delta'=\frac{\delta}{c\log n+1}$ and $\eta(G)=O(\log n)$, the graph $G'$ is a $(1\pm\delta)$-ASC of $G$ w.r.t. $\partial(G)$.

\paragraph*{Handling Edge Insertions.} We now describe the \textsc{Insert} operation. Let us consider the insertion of an edge $e=(u,v)$ of weight $w$. We maintain a stack $Q$, which is initially set to empty. We then update the root node by adding $(u,v)$ with weight $w$ to $G$, and push $G$ onto $Q$. During the traversal of $\mathcal{T}(G)$, our procedure maintains two pointers that point to the current node $H$ (initially set to $G$) and a node $N$ (if any exists) that represents the node for which $u$ and $v$ belong to different children of $N$, respectively. As long as we have not found such a node $N$, and the current node $H$ is not a leaf, we proceed as follows. 

We examine the child of $H$ that contains both $u$ and $v$ (if there is more than one, then we just pick one of them). If $u$ and $v$ belong to the same child, say $\child(H)$, then we add this edge to $\child(H)$ and update the current node $H$ to $\child(H)$. We then push $H$ onto $Q$. If, however, $u$ and $v$ belong to different children, then we set $N$ to be the current node $H$ and add the edge $(u,v)$ to $\buffer(N)$, since $u$ and $v$ cannot appear together in the nodes of the lower levels. At this point, this forces $u$ and $v$ to become boundary vertices in $N$ and all other nodes descending from $N$ that contain either $u$ or $v$. We handle this by making use of the $\textsc{AddBoundary}()$ procedure, depicted in Algorithm~\ref{alg: addBoundary}. Finally, we recompute the ASCs of the affected nodes in a bottom-up fashion using the stack $Q$ (as shown in Algorithm~\ref{alg: updateApproxSchur}). This procedure is summarized in Algorithm~\ref{alg: insert}. We remark that for simplicity, we let $Q.\textsc{Push}(H)$ denote the event of pushing the pointer to $H$ to the stack $Q$, for any node $H$.

\begin{algorithm2e}[htb!]
\caption{\textsc{UpdateApproxSchur}$(\textsc{Stack } Q)$}
\label{alg: updateApproxSchur}

\While{$Q \neq \emptyset$} 
{
  Set $H \gets Q.\textsc{Pull}()$. \\
  Set $H' \gets \textsc{ApproxSchurNode}(H,\partial(H),\varepsilon)$. \\
}

\end{algorithm2e}

\begin{algorithm2e}[htb!]
\caption{\textsc{Insert}$(u,v,w)$}
\label{alg: insert}
Let $Q$ be an initially empty stack. \\
Set $E(G) \gets E(G) \cup \{(u,v)\}$,  $Q.\textsc{Push}(G)$, $H \gets G$ and  $N \gets \nil$. \\
\While {$N = \nil$ and $H$ is a non-leaf}
{
	\eIf{there exists a child of $H$ that contains both $u$ and $v$}
	{
		Let $\child(H)$ denote any such a child. \\
		Set $E\left(\child(H)\right) \gets E\left(\child(H)\right) \cup \{(u,v)\}$. \\
		Set $H \gets \child(H)$. \\
		$Q.\textsc{Push}(H)$. \\
	}
	{
	    Set $N \gets H$. \\
        Set $\buffer(N) \gets \buffer(N) \cup \{(u,v)\}$. \\
        \textsc{AddBoundary}$(u,N)$, \textsc{AddBoundary}$(v,N)$. \\
	}
}

\nonl\texttt{// Update the ASCs of the nodes in $Q$} \\
\textsc{UpdateApproxSchur}$(Q)$.
\end{algorithm2e}

After the pre-processing step and after each insertion/deletion of an edge, our augmented separator tree $\mathcal{T}(G)$ satisfies the following invariant.

\begin{invariant} For every edge $e$ in the current graph $G$, exactly one of the following two holds:
\begin{itemize}
\item there is a leaf node $H \in \mathcal{T}(G)$ such that $e \in E(H)$,
\item there is an internal node $H \in \mathcal{T}(G)$ such that $e \in \emph{X}(H)$.
\end{itemize}
\end{invariant}


The following lemma guarantees that the updated graph $G'$ (i.e., the sparsifier of the root node $G$) is good approximation to the Schur complement of $G$ with respect to the boundary, after the execution of $\textsc{Insert}(u,v)$ in Algorithm~\ref{alg: insert}. 

\begin{lemma}\label{lemma:insertion} Let $G'$ be the updated sparsifier of the root node $G$, after the insertion of edge $(u,v)$. Then $G'$ is an $(1 \pm \delta)$-\emph{ASC} of $G$ with respected to $\partial(G)$.
\end{lemma}
\begin{proof}
	We proceed inductively as in the proof of Lemma~\ref{lem: dsCorrectness} and show that for any node $H$, the corresponding sparsifier $H'$ is an $(1\pm \delta')^{\eta(H)+1}$-ASC of $H$ with respect to $\partial(H)$. Since the base case remains the same, let us consider a non-leaf node $H$. If $\buffer(H) = \emptyset$, then the correctness follows from the inductive step of Lemma~\ref{lem: dsCorrectness}. However, $\buffer(H) \neq \emptyset$ implies that $H \neq \child_1(H) \oplus \child_2(H)$. This is because $H$ is the last node for the edges of $\buffer(H)$ whose endpoints were contained in the same node in $\mathcal{T}(G)$. Recall that the endpoints of all the edges in $\buffer(H)$ were declared boundary vertices for $H$ and all descendants containing them. Thus we have that
	\[
	H = \left(\child_1(H) \oplus \child_2(H)\right) \cup \buffer(H).
	\]
	By induction hypothesis, it follows that $c_i(H)'$ is an $(1\pm \delta')^{\eta(\child_i(H)) + 1}$-ASC of $c_i(H)$, for $i=1,2$. Using Lemma~\ref{lem: approxMerge} and since $\eta(\child_i(H)) + 1 = \eta(H)$, for $i=1,2$, we get that $R := \child_1(H)' \oplus_{\phi} \child_2(H)'$ is an $(1\pm \delta')^{\eta(H)}$-ASC of $H \setminus \buffer(H)$ with respect to $V(R):= \partial(\child_1(H)) \cup \partial(\child_2(H))$. First, since $V(R) \supseteq V(\buffer(H))$ by construction, Lemma~\ref{lem: approxMerge} implies that $R' := R \cup \buffer(H)$ is an $(1\pm \delta')^{\eta(H)}$-ASC of $(H \setminus \buffer(H)) \cup \buffer(H) = H$ with respect to $V(R)$. Second, since $V(R) \supseteq \partial(H)$ by property (4) of $\mathcal{T}(G)$ and by Lemma~\ref{lem: ApproxSchur}, it follows that $H'$ is an $(1 \pm \delta')$-ASC of $R'$ with respect to $\partial(H)$. Finally, applying Lemma~\ref{lem: multiComb} on $R'$ and $H'$ we get that $H'$ is an $(1 \pm \delta')^{\eta(H) + 1}$-ASC of $H$. The statement of the lemma then follows from the facts that $\delta'=\frac{\delta}{c\log n+1}$ and $\eta(G)=O(\log n)$.
\end{proof}

For the running time of $\textsc{Insert}(u,v,w)$, we distinguished two cases. 

First, suppose that the insertion of the edge $(u,v)$ does not trigger a re-computation of the data-structure. Note that the stack $Q$ (in Algorithm~\ref{alg: insert}) contains all nodes in the path starting from the root node $G$, and then repeatedly choosing \emph{exactly one} child of the current node that contains both $u$ and $v$, until the node $N$ is reached. Since the height of $\mathcal{T}(G)$ is $O(\log n)$, it follows that $\abs{Q} \leq O(\log n)$. Additionally, by Lemma~\ref{lemm: computeAscNode}, the time to re-compute an ASC of any node is bounded by $\tilde{O}(\sqrt{n}/\delta^{2})$. Thus we get that the time needed to update the ASCs of the nodes in $Q$ is $\tilde{O}(\sqrt{n}/\delta^{2})$. As we will shortly argue, the running time of $\textsc{AddBoundary}()$ is also bounded by $\tilde{O}(\sqrt{n}/\delta^{2})$. Combining the above, we get that the running time of $\textsc{Insert}(u,v)$ is $\tilde{O}(\sqrt{n}/\delta^{2})$.

Second, suppose that the edge $(u,v)$ triggers a re-computation of the data-structure. Then by Lemma~\ref{lem: DSrunningTime}, we recompute $\mathcal{D}(G, \delta)$ in $\tilde{O}(n/\delta^{2})$ time. Since we recompute that data-structure every $\Theta(\sqrt{n})$ insertions, the amortized update time per insertion is $\tilde{O}(\sqrt{n}/\delta^{2})$. The above bounds combined give that the amortized time per edge insertion is bounded by $\tilde{O}(\sqrt{n}/\delta^{2})$. This bound can be made worst-case by keeping two copies of the data structure and performing periodical rebuilds.

\paragraph*{Handling Terminal Additions to the Boundary.} We now describe the \textsc{AddTerminal}$(u)$ operation. 
It is implemented by simply invoking \textsc{AddBoundary}$(u,G)$, where $G$ is the root of $\mathcal{T}(G)$. For the procedure $\textsc{AddBoundary}(u,H)$, we maintain a stack $Q$, which is initially set to empty. As long as the current $H$ is a node in $\mathcal{T}(G)$, we first check whether $u \in \partial(H)$. If this is the case, then we simply do nothing as the ASC $H'$ of $H$ with respect to $\partial(H)$ contains $u$. Otherwise, we add $u$ to $\partial(H)$, and push the node $H$ to $Q$. Next, if $H$ is not a leaf-node, let $\child(H)$ be the \emph{unique} child that contains $u$. We then set $\child(H)$ to be our current node $H$ and perform the same steps as above, until we reach some leaf-node, in which case we set $H$ to $\nil$. Finally, we recompute the ASCs of the affected nodes in a bottom-up fashion using the stack $Q$. This procedure is summarized in Algorithm~\ref{alg: addBoundary}. 

The correctness of this procedure can be shown similarly to the correctness of $\textsc{Insert}()$. For the running time, the crucial observation is that if $u \not \in \partial(H)$, for some non-leaf node $H$, then by property (2) of $\mathcal{T}(G)$, it follows that $u$ is assigned to an \emph{unique} child of $H$. Thus, in the worst-case, the stack $Q$ contains all the nodes in the path between $H$ and some leaf-node. Note that $\abs{Q} = O(\log n)$ and by Lemma~\ref{lemm: computeAscNode}, time to re-compute an ASC of any node is $\tilde{O}(\sqrt{n}/\delta^{2})$. Combining the above, we get that the running time of $\textsc{AddBoundary}(u,H)$ is $\tilde{O}(\sqrt{n}/\delta^{2})$.

\begin{algorithm2e}[htb!]
\caption{\textsc{AddBoundary}$(u,v,w)$}
\label{alg: addBoundary}
Let $Q$ be an initially empty stack. 
\While {$N = \nil$}
{
	\If{$u \not \in \partial(H)$}
	{
		Set $\partial(H) \gets \partial(H) \cup \{u\}$. \\
	    $Q.\textsc{Push}(H)$. \\
		\If {$H$ is a non-leaf} 
		{
			Let $c(H)$ be the \emph{unique} child that contains $u$. \\
			Set $H \gets \child(H)$. \\
		}
	}
	
	\If{$H$ is a leaf}
	{
		Set $H \gets \nil$. \\
	}
}

\nonl\texttt{// Update the ASCs of the nodes in $Q$} \\
\textsc{UpdateApproxSchur}$(Q)$.

\end{algorithm2e}

\paragraph*{Handling Edge Deletions.} We now describe the \textsc{Delete} operation. Let us consider the deletion of an edge $e = (x,y)$. Our procedure is symmetric to the \textsc{Insert}() operation. We maintain a stack $Q$, which is initially set to empty. We then update the root node by deleting $(u,v)$ from $G$, and pushing $G$ onto $Q$. During the traversal of $\mathcal{T}(G)$, our procedure maintains the current node $H$ (initially set to $G$) and determines the node $N$ that represent the lowest-level node in $\mathcal{T}(G)$ that contains the edge $(u,v)$. Note that $N$ is not necessarily a leaf-node. As long as we have not found such a node we proceed as follows. 

We examine the \emph{unique} child of $H$ that contains the edge $(u,v)$ (by property (2) of $\mathcal{T}(G)$).  If there exists such a child $c(H)$, then we delete $(u,v)$ from $c(H)$ and update the current node $H$ to $c(H)$. We then push $H$ to $Q$. If, however, such a child does not exist, then we set $N$ to be the current node $H$. Next, if $N$ is a non-leaf node, we remove the edge $(u,v)$ from $\buffer(N)$. Finally, we recompute the ASCs of the affected nodes in a bottom-up fashion using the stack $Q$. This procedure is summarized in Algorithm~\ref{alg: delete}.

Similarly to the \textsc{Insert}$()$ operation, we can show  that the worst-case running time of \textsc{Delete}$(u,v)$ operation is $\tilde{O}(\sqrt{n}/\delta^{2})$.

Finally, recall that we set $\gamma=1/n^3$ as the error probability of \textsc{ApproxSchur} from Lemma~\ref{lem: ApproxSchur}. This will guarantee that throughout all updates, our algorithm succeeds with probability at least $1-O(n)\cdot \frac{1}{n^3}\geq 1- O(\frac{1}{n^2})$ as the total number of nodes in $\mathcal{T}(G)$ is $O(\sqrt{n})$, each update involves recomputation of the ASCs of $O(\log n)$ nodes and our algorithm recomputes the data structure every $\Theta(\sqrt{n})$ operations.

\paragraph{Remark.} We can easily generalize the above framework to $O(\sqrt{n})$-separable graphs for which the separator can be computed in $s(n)$ time, since the only place we need such computation is to initialize or re-compute the data structure $\mathcal{D}(G,\delta)$ (after every $\Theta(\sqrt{n})$ operations). This implies that the update time will become $\tilde{O}((s(n)+n/\delta^{2})/\sqrt{n})$ and the query time remains the same as before.  


\begin{algorithm2e}[htb!]
\caption{\textsc{Delete}$(u,v)$}
\label{alg: delete}
Let $Q$ be an initially empty stack. \\
Set $E(G) \gets E(G) \setminus \{(u,v)\}$,  $Q.\textsc{Push}(G)$, $H \gets G$ and  $N \gets \nil$. \\
\While {$N = \nil$}
{
	\eIf{If there exists a (\emph{unique}) child $\child(H)$ of $H$ that contains $(u,v)$}
	{
		$E\left(\child(H)\right) \gets E\left(\child(H)\right) \setminus \{(u,v)\}$. \\
		Set $H \gets \child(H)$. \\
		$Q.\textsc{Push}(H)$. \\
	}
	{
	    Set $N \gets H$. \\
	}
	
	\If{$N$ is a non-leaf} 
	{ 
	  $\buffer(N) \gets \buffer(N) \setminus \{(u,v)\}$. \\
	}	
}

\nonl\texttt{// Update the ASCs of the nodes in $Q$} \\
\textsc{UpdateApproxSchur}$(Q)$.

\end{algorithm2e}

\subsection{Extension to Dynamic All-Pairs Effective Resistance}

We next explain how to use a dynamic ASC algorithm to obtain a fully-dynamic algorithm for maintaining an $(1+\varepsilon)$-approximation to all-pairs (resp., single-pair) effective resistance(s) in a $O(\sqrt{n})$-separable graph $G$ and prove Theorem~\ref{thm:upperbound}. The data-structure support the operations $\textsc{Insert}(u,v,r)$, $\textsc{Delete}(u,v)$, and  \textsc{EffectiveResistance}$(s,t)$ as defined in the beginning of the paper. 

Our dynamic effective resistance algorithm uses the above dynamic algorithm for maintaining a $(1\pm \delta)$-ASC as a subroutine. Formally, to maintain $(1+\varepsilon)$-approximations of effective resistances, we will invoke the dynamic ASC algorithm with parameters $\delta = \varepsilon/4$.
To answer the queries of the effective resistance of any two given vertices, we use the following result due to Durfee et al.~\cite{DurfeeKPRS16}.

\begin{theorem} \label{thm: spielmanTeng}
	Fix $\delta \in (0,1/2)$ and let $G=(V,E , \vect{w})$ be a weighted graph with two distinguished vertices $s,t \in V$. There is an algorithm \textsc{EstimateEffRes}$(G,s,t)$ that computes a value $\psi$ such that
	\[
	(1-\delta) R_G(s,t) \leq \psi \leq (1+\delta)R_G(s,t), 
	\]
	in time $\tilde{O}(m + n/\delta^{2})$ with probability at least $1-n^c$ for some constant $c\geq 1$.
\end{theorem}

For simplicity, we focus on the case that the separator of the separable graph can be computed in $\tilde{O}(n)$ time. The algorithm and analysis can be easily generalized to handle the case when the computation time for separator is $s(n,m)$, by the same argument as before.

We now describe the query operation. We first consider how to maintain all-pairs effective resistances. Given $s$ and $t$, we start by calling  $\textsc{AddTerminal}(s)$ and $\textsc{AddTerminal}(t)$ from the dynamic ASC data-structure. This ensures that both $s$ and $t$ are boundary nodes at the root node $G$ (if they were not previously). Thus we obtain a $(1\pm \delta)$-ASC, denoted as $G'$, of the root node $G$ with respect to $\partial(G)$ and run on $G'$ a nearly linear time algorithm for estimating the $s-t$ effective resistance (see Theorem~\ref{thm: spielmanTeng}). Let $\psi$ denote such an estimate. This procedure is summarized in Algorithm~\ref{alg: effectiveResistance}.

\begin{algorithm2e}[htb!]
\caption{\textsc{EffectiveResistance}$(s,t)$}
\label{alg: effectiveResistance}
\textsc{AddTerminal}$(s)$, \textsc{AddTerminal}$(t)$.  \\
Let $G'$ be the ASC of the root node $G$ with respect to $\partial(G)$. \\
Set $\psi \gets$ \textsc{EstimateEffRes}$(G',s,t)$. \\
Return $\psi$. \\
\end{algorithm2e}

For the correctness, by Lemma~\ref{lemma:exact_schur}, we have that $G'$ preserves all-pairs effective resistances among vertices in $\partial(G)$ of $G$ up to an $1/(1 \pm \delta) \approx (1 \pm 2\delta)$ factor. Since we ensured that $s$ and $t$ are included in $\partial(G)$, the $s-t$ effective resistance is approximated within the same factor. By Theorem~\ref{thm: spielmanTeng}, it follows that the estimate $\psi$ approximates the effective resistance between $s$ and $t$ in $G'$, up to a $(1\pm \delta)$ factor. Combining the above guarantees, we get $\psi$ gives an $(1 \pm 2\delta)(1\pm \delta) \leq (1 \pm \varepsilon)$-approximation to $R_G(s,t)$, by the choice of $\delta$.

Once the query is answered, we then undo all the changes that we have performed in $\mathcal{T}(G)$ i.e., we bring the data-structure to its state before the query operation. This ensures that the number of terminals at the root node $G$ does not accumulate over a large sequence of query operations.

For the running time, first recall that each $\textsc{AddTerminal}()$ operation can be implemented in $\tilde{O}(\sqrt{n}/\delta^{2})$. Now, as $\abs{V(G')} \leq O(\sqrt{n})$ and $\abs{E(G')} \leq \tilde{O}(\sqrt{n}/\delta^{2})$, by Theorem~\ref{thm: spielmanTeng} it follows that estimate $\psi$ can be computed in $\tilde{O}(\sqrt{n}/\delta^{2})$ time. Combining the time bounds we get that that the worst-case time to answer an $\textsc{EffectiveResistance}(s,t)$ query is $\tilde{O}(\sqrt{n}/ \delta^{2})$. Finally, note that in the same time bound, we can also undo all the changes we have made.

For the single-pair $s-t$ effective resistance, the two vertices $s,t$ are fixed throughout all the operations. For each edge insertion or deletion, we first update the data structure in the same way as for the all-pairs version, and then we compute the $s-t$ effective resistance $R_G(s,t)$ and store the answer. For the query for $R_G(s,t)$, we simply report the stored answer. The update time is still $\tilde{O}(\sqrt{n}/\delta^2)$, while the query time is only $O(1)$.







\section{Lower Bounds for Partially Dynamic Effective Resistances} \label{sec: lowerBound}

\subsection{A Lower Bound for $O(\sqrt{n})$-Separable Graphs}\label{sec:proof_lower_separable}
In this section, we prove a conditional lower bound for incrementally or decrementally maintaining the $s-t$ effective resistance in $O(\sqrt{n})$-separable graphs and give the proof of Theorem~\ref{thm:lowerbound_separable}. Our proof actually holds for any algorithm that maintains a $(1+O(\frac{1}{n^{36}}))$-approximation of $s-t$ effective resistance. 

We first consider the incremental case, in which only edge insertions are allowed. 

\paragraph*{The reduction.} We reduce the $\vect{u}\mat{M}\vect{v}$ problem~(see Definition \ref{def:umv_problem}) with parameters $n_1=n_2:=n_0$ to the $s-t$ effective resistance problem as follows. Let $\mat{M}$ be the $n_0\times n_0$ Boolean matrix of the $\vect{u}\mat{M}\vect{v}$ problem. Let $n= n_0^2+2n_0+2$. Let $\kappa = 3(n-1)^6$.

Given the matrix $\mat{M}$, we construct a graph $G_\mat{M}=(V_\mat{M}, E)$ as follows. 
\begin{itemize}
	\item  For each pair $1\leq i,j\leq n_0$, we create two vertices $a_{ij}$ and $b_{ij}$, and add an edge $(a_{ij},b_{ij})$ if and only if $M_{ij}=1$. 
	\item For each row $i$, we create a vertex $u_i$ and add edge $(u_i,a_{ik})$ for each $1\leq k\leq n_0$. For each column $j$, we create a vertex $v_j$ and add edge $(v_j,b_{kj})$ for each $1\leq k\leq n_0$. 
\end{itemize}		
This finishes the definition of $G_\mat{M}$. Note that $V_\mat{M}=\{a_{ij},b_{ij}, 1\leq i,j\leq n_0 \}\cup\{u_i,1\leq i\leq n_0 \} \cup\{v_j, 1\leq j\leq n_0\}$. For any vertex $x\in V_\mat{M}$, let $\deg_{G_\mat{M}}(x)$ denote the degree of $x$ in $G_\mat{M}$. 


Now we add two new vertices $t$ and $s$ to $G_\mat{M}$. For any $x\in \{a_{ij},b_{ij}, 1\leq i,j\leq n_0 \}$, add an edge $(s,x)$ with weight $\kappa-\deg_{G_\mat{M}}(x)$. Denote the resulting graph by $G$ and note that $G$ contains $|V_\mat{M}\cup\{s,t\}|=n_0^2+2n_0+2 = n$ vertices.

Assume that $G$ is started in a dynamic effective resistance data structure. We also maintain a number of counters in the data structure. More specifically, we initialize a global counter $Y:=0$. For each vertex $x\in \{u_i, 1\leq i\leq n_0 \}\cup \{v_j, 1\leq j\leq n_0\}$, we maintain a counter $c(x)$ which is initialized to be $0$. 
We now explain how we use this data structure to determine $\vect{u}\mat{M}\vect{v}$.

\begin{itemize}
	\item Once $\vect{u}$ arrives, for any $i$ such that $\vect{u}_i=1$, we insert an edge $(t, u_i)$ with weight $1$, increase $Y$ and $c(u_i)$ by $1$. 
	\item Once $\vect{v}$ arrives, for any $j$ such that $\vect{v}_j=1$, we insert an edge $(t, v_j)$ with weight $1$, increase $Y$ and $c(v_j)$ by $1$. 
	\item Insert an edge $(s,t)$ with weight $\kappa - Y$. For each vertex $x\in \{u_i, 1\leq i\le n_0 \}\cup\{v_j, 1\leq j\leq n_0 \}$, insert an edge $(s,x)$ with weight $\kappa-c(x)-\deg_{G_\mat{M}}(x)$.
	\item We perform one effective resistance query $\textsc{EffectiveResistance}(s,t)$ to obtain the (approximate) $s-t$ effective resistance in the final graph. Let $\lambda=\textsc{EffectiveResistance}(s,t)$. If $\lambda\leq \frac{1}{\kappa}+\frac{Y}{\kappa^3}+\frac{Y(n_0+1)}{\kappa^5} -\frac{1}{\kappa^6}$, then return $1$; otherwise, return $0$. 
\end{itemize}

\paragraph*{Analysis.} Note that throughout the whole sequence of updates (which are only edge insertions) and queries, the dynamic graph $G$ is always $O(\sqrt{n})$-separable, since the set $S:=\{u_1,\cdots, u_{n_0}\}\cup\{v_1,\cdots,v_{n_0} \}\cup\{s,t\}$ is a balanced separator of size $O(\sqrt{n})$.

We have the following lemma that shows an important property of our reduction. The proof of the lemma is deferred to the end of this section.
\begin{lemma}\label{lemma:separable_LB}
	For $\kappa = 3(n-1)^6$, assume that $\textsc{EffectiveResistance}(s,t)$ returns an $(1+\frac{1}{\kappa^{6}})$-approximation of the $s-t$ effective resistance in the final graph $G$. Then the following holds:
	\begin{itemize}
		\item  If $\vect{u}\mat{M}\vect{v}=1$, then $\lambda \leq \frac{1}{\kappa}+\frac{Y}{\kappa^3}+\frac{Y(n_0+1)}{\kappa^5} -\frac{1}{\kappa^6}$;
		\item  If $\vect{u}\mat{M}\vect{v}=0$, then $\lambda >\frac{1}{\kappa}+\frac{Y}{\kappa^3}+\frac{Y(n_0+1)}{\kappa^5} -\frac{1}{\kappa^6}$.
	\end{itemize}
\end{lemma}

Note that by the above lemma, the $\vect{u}\mat{M}\vect{v}$ problem can be solved according to our estimator $\lambda$. Thus, the lower bound for the incremental setting in Theorem~\ref{thm:lowerbound_separable} follows by Theorem~\ref{thm:umv_hard} and by noting that the total number of updates is $O(n_0)=O(\sqrt{n})$ and the total number of queries is $1$.

In the following we prove Lemma~\ref{lemma:separable_LB}. The proof is based on a connection between the $5$-length cycle detection problem and the effective resistance problem.
\begin{proof}[Proof of Lemma~\ref{lemma:separable_LB}]
	Let $G$ denote the final graph of our reduction. Let $H:=G[V_\mat{M}\cup\{t\}]$ denote the subgraph induced by vertex set $V_\mat{M}\cup\{t\}$. We observe that in the graph $H$, there is a cycle of length $5$ containing vertex $t$ if and only if $\vect{u}\mat{M}\vect{v}=1$. 
	
	On the other hand, we can use our estimator $\lambda$ to distinguish if $H$ contains a $5$-length cycle incident to $t$ or not. 
	We let $\mat{A}\in \mathbb{R}^{(n-1)\times (n-1)}$ denote the adjacency matrix of the graph $H$. Note that all entries in $A$ are either $1$ or $0$.
	
	The first claim relates the $5$-length cycle detection to the trace of a matrix related to $\mat{A}$. Recall that we let $X_{uv}$ denote the entry of matrix $X$ with row index corresponding to vertex $u$ and column index corresponding to vertex $v$.
	\begin{claim}\label{claim:cycle_trace}
		Let $\mat{B}=\kappa\cdot \mat{I} - \mat{A}$. If $H$ contains a $5$-length cycle incident to $t$, then $(\mat{B}^{-1})_{tt} \leq \frac{1}{\kappa}+\frac{Y}{\kappa^3}+\frac{Y(n_0+1)}{\kappa^5}-\frac{1.1}{\kappa^6}$. If $H$ does not contain a $5$-length cycle incident to $t$, then $(\mat{B}^{-1})_{tt} \geq \frac{1}{\kappa}+\frac{Y}{\kappa^3}+\frac{Y(n_0+1)}{\kappa^5} -\frac{0.9}{\kappa^6}$.
	\end{claim}
	\begin{proof}
		First we note that $B$ is invertible, as it is strictly symmetric
		diagonally dominant. Furthermore, it holds that 
		$\kappa\cdot\mat{B}^{-1}=(I-\frac{1}{\kappa}\cdot \mat{A})^{-1}$ and thus by the Neumann series expansion, we have \[\kappa\cdot \mat{B}^{-1} =(I-\frac{1}{\kappa}\cdot \mat{A})^{-1}= \sum_{i=0}^\infty (-\frac{1}{\kappa})^i \cdot \mat{A}^i.\]
		This further implies that 
		\begin{eqnarray}
		(\kappa\cdot \mat{B}^{-1})_{tt} 
		= \1_t^T(\sum_{i=0}^\infty (-\frac{1}{\kappa})^i \cdot \mat{A}^i)\1_t=\sum_{i=0}^\infty (-\frac{1}{\kappa})^i \cdot\1_t^T(\mat{A}^i)\1_t = \sum_{i=0}^\infty (-\frac{1}{\kappa})^i \cdot (\mat{A}^i)_{tt}.	\label{eqn:inverse_B}
		\end{eqnarray}

		Now observe that since $\kappa=3 (n-1)^6$, the first six terms of the above power series dominate. More precisely, note that $(\mat{A}^i)_{tt}$ is the number of $i$-length paths from $t$ to $t$, which is at most $(n-1)^i$. Thus 
		\begin{eqnarray*}
			\sum_{i=6}^\infty \abs{(-\frac{1}{\kappa})^i \cdot (\mat{A}^i)_{tt}}
			\leq 
			\sum_{i=6}^\infty \frac{1}{\kappa^i} (\mat{A}^i)_{tt} 
			\leq 
			\sum_{i=6}^\infty \frac{1}{\kappa^i} (n-1)^i
			\leq
			\frac{0.9}{\kappa^5}.
		\end{eqnarray*}
		
		
		
		Now observe that $(\mat{A}^0)_{tt}=\mat{I}_{tt}=1$; that $\mat{A}_{tt}=0$ since $H$ is a simple graph; that $(\mat{A}^2)_{tt}=\deg_H(t)=Y$, where the last equation follows from the definition of $Y$; that $(\mat{A}^3)_{tt}=0$ since there is no triangle containing $t$; and that $(\mat{A}^4)_{tt}=\sum_{w: (w,t)\in E}\sum_{x:(x,w)\in E}1=\sum_{w: (w,t)\in E} \deg_{G_\mat{M}}(w) = \det_H(t)\cdot (n_0+1) = Y(n_0+1)$. Therefore, 
		\begin{itemize}
			\item If $H$ contains a $5$-length cycle incident to $t$, then $(\mat{A}^5)_{tt}\geq 2$, and thus
			\begin{eqnarray*}
				(\kappa\cdot \mat{B}^{-1})_{tt} \leq 1+\frac{Y}{\kappa^2} + \frac{Y(n_0+1)}{\kappa^4} - \frac{2}{\kappa^5} + \frac{0.9}{\kappa^5} = 1+\frac{Y}{\kappa^2}+ \frac{Y(n_0+1)}{\kappa^4} -\frac{1.1}{\kappa^5} 
			\end{eqnarray*}
			\item If $H$ has no $5$-length cycle incident to $t$, then $(\mat{A}^5)_{tt}=0$, and thus
			\begin{eqnarray*}
				(\kappa\cdot \mat{B}^{-1})_{tt} \geq 1+\frac{Y}{\kappa^2} + \frac{Y(n_0+1)}{\kappa^4} - \frac{0.9}{\kappa^5}
			\end{eqnarray*}
		\end{itemize}
		This completes the proof of the claim.
	\end{proof}
	
	The following claim relates $s-t$ effective resistance to $\mat{B}^{-1}$. The proof almost follows from Lemma 23 in~\cite{MNSUW17:spectrum}, while we include a proof in Appendix~\ref{app: lowBound} for the sake of completeness.
	\begin{claim}\label{claim:trace}
		Let $\Lambda=\mathcal{E}_G(s,t)$ and $\mat{B}=\kappa\cdot \mat{I} - \mat{A}$. Then it holds that 
		$\Lambda=(\mat{B}^{-1})_{tt}.$
	\end{claim}
	
	Finally, by the above two claims, if $\vect{u}\mat{M}\vect{v}=1$, then $H$ contains a $5$-length cycle incident to $t$, and thus $\Lambda=(\mat{B}^{-1})_{tt} \leq \frac{1}{\kappa}+\frac{Y}{\kappa^3}+\frac{Y(n_0+1)}{\kappa^5}-\frac{1.1}{\kappa^6}$; if $\vect{u}\mat{M}\vect{v}=0$, then $H$ does not contain any $5$-length cycle incident to $t$, and thus $\Lambda=(\mat{B}^{-1})_{tt} \geq \frac{1}{\kappa}+\frac{Y}{\kappa^3}+\frac{Y(n_0+1)}{\kappa^5} -\frac{0.9}{\kappa^6}$. The statement of the lemma then follows by the fact that $\lambda$ is a $(1+\frac{1}{\kappa^6})$-approximation of $\Lambda$, and that $\frac{1}{\kappa^6}(\frac{1}{\kappa}+\frac{Y}{\kappa^3}+\frac{Y(n_0+1)}{\kappa^5}-\frac{0.9}{\kappa^6})<\frac{0.1}{\kappa^6}$.
\end{proof}

For the lower bound for the decremental setting, 
we start with a graph where $t$ is initially connected to $s$ with weight $\kappa-2n_0$ and to all vertices $x\in\{u_i,1\leq i\leq n_0 \} \cup\{v_j,1\leq j\leq n_0 \}$ with weights $\kappa-1-\deg_{G_\mat{M}}(x)$. When the vectors $\vect{u},\vect{v}$ arrive, we need to increase the weights of some edges $(s,x)$ and $(s,t)$ depending if the corresponding entry of $\vect{u},\vect{v}$ is $1$ or $0$, so that every vertex in $G$ has the same weighted degree $\kappa$. We omit further details here.

\subsection{A Lower Bound for General Graphs}
In this section, we prove Theorem~\ref{thm:lowerbound_general}, which gives a lower bound for incremental and decremental $s-t$ effective resistance problem in general graphs. 

\begin{proof}[Proof of Theorem~\ref{thm:lowerbound_general}]
	We only consider here the incremental setting, where only edge insertions are allowed. For the decremental setting, the correctness follows from a similar construction and similar arguments for decremental lower bound in the proof of Theorem~\ref{thm:lowerbound_separable}. 
	
	We reduce the $\vect{u}\mat{M}\vect{v}$ problem with parameters $n_1=n_2:=n_0$ to the $s-t$ effective resistance problem as follows. Let $\mat{M}$ be the $n_0\times n_0$ Boolean matrix of the $\vect{u}\mat{M}\vect{v}$ problem. Let $n=2n_0+2$ and let $\kappa = 3(n-1)^5$.
	
	We first create a bipartite graph $G_\mat{M}=((R,C),E)$ where $R=(r_1,\cdots,r_{n_0})$ and $C=(c_1,\cdots, c_{n_0})$ corresponding to the rows and columns of $\mat{M}$, respectively. We add an edge $(r_i, c_j)$ in $E$ iff $\mat{M}_{ij}=1$. This finishes the definition of $G_\mat{M}$. For each vertex $x\in R\cup C$, let $\deg_{G_\mat{M}}(x)$ denote the degree of vertex $x$ in $G_\mat{M}$.
	
	Now we add tow new vertices $s,t$ to $G_\mat{M}$. Denote the resulting graph by $G$ and note that $G$ contains $|R\cup C\cup\{s,t\}|=2n_0+2$ vertices. 
	
	Assume that $G$ is started in a dynamic effective resistance data structure. We also initialize a global counter $Y$ to be $0$ and for each vertex $x\in R\cup C$, we initialize a counter $c(x)$ to be $0$. We now explain how we use this data structure to determine $\vect{u}\mat{M}\vect{v}$.
	
	\begin{itemize}
		\item Once $\vect{u}$ arrives, for any $i$ such that $\vect{u}_i=1$, we insert an edge $(t, r_i)$ with weight $1$, and increase $Y$ and $c(r_i)$ by $1$. 
		\item Once $\vect{v}$ arrives, for any $j$ such that $\vect{v}_j=1$, we insert an edge $(t, c_j)$ with weight $1$, and increase $Y$ and $c(c_j)$ by $1$. 
		\item Insert an edge $(s,t)$ with weight $\kappa-Y$. For each $x\in V_\mat{M}$, insert an edge $(s,x)$ with weight $\kappa-c(x)-\deg_{G_\mat{M}}(x)$. 
		\item We perform one effective resistance query $\textsc{EffectiveResistance}(s,t)$ to obtain the (approximate) $s-t$ effective resistance in the final graph. Let $\lambda=\textsc{EffectiveResistance}(s,t)$. If $\lambda\leq \frac{1}{\kappa}+\frac{Y}{\kappa^3}-\frac{1}{\kappa^4}$, then return $1$; otherwise, return $0$.
	\end{itemize}
	
	We have the following lemma similar to Lemma~\ref{lemma:separable_LB}. 
	\begin{lemma}\label{lemma:correctness_LB}
		For $\kappa=3(n-1)^5$, assume that $\textsc{EffectiveResistance}(s,t)$ returns a $(1+\frac{1}{\kappa^{4}})$-approximation of the $s-t$ effective resistance in the final graph $G$. Then the following holds:
		\begin{itemize}
			\item If $\vect{u}\mat{M}\vect{v}=1$, then $\lambda \leq \frac{1}{\kappa}+\frac{Y}{\kappa^3}-\frac{1}{\kappa^4}$;
			\item If $\vect{u}\mat{M}\vect{v}=0$, then $\lambda >\frac{1}{\kappa}+\frac{Y}{\kappa^3}-\frac{1}{\kappa^4}$.
		\end{itemize}
	\end{lemma}
	Given the above Lemma, we can then solve the $\vect{u}\mat{M}\vect{v}$ problem according to the value of our estimator $\lambda$. Thus, the statement of the theorem follows by noting that the total number of updates is $O(n_0)=O(n)$ and the total number of queries is $1$, and by Theorem~\ref{thm:umv_hard}. Now we give a sketch of the proof of the above lemma.
	\begin{proof}[Proof Sketch of Lemma~\ref{lemma:correctness_LB}]
		The proof is almost the same as the proof of Lemma~\ref{lemma:separable_LB}. Here we point out the main difference. Let $G$ denote the final graph of our reduction. Let $H:=G[R\cup C\cup\{t\}]$ denote the subgraph induced by vertex set $R\cup C\cup\{t\}$. We observe that in the graph $H$, there is a triangle incident to vertex $t$ iff $\vect{u}\mat{M}\vect{v}=1$. Now we use our estimator $\lambda$ to distinguish if $H$ contains a triangle incident to $t$ or not.
		
		We let $\mat{A}\in \mathbb{R}^{(n-1)\times (n-1)}$ denote the adjacency matrix of the graph $H$. Note that all entries in $A$ are either $1$ or $0$. Let $\mat{B}=\kappa\cdot \mat{I} - \mat{A}$. Again, by the Neumann series expansion of $\mat{B}^{-1}$, we could derive the same expression of $(\kappa\cdot \mat{B}^{-1})_{tt}$ as Equation~\ref{eqn:inverse_B}, that is 
		\begin{eqnarray*}
			(\kappa\cdot \mat{B}^{-1})_{tt} = \sum_{i=0}^\infty (-\frac{1}{\kappa})^i \cdot (\mat{A}^i)_{tt}.	
		\end{eqnarray*}
		
		Now observe that since $\kappa=3(n-1)^5$, the first four terms of the above power series dominate. More precisely, by the fact that $(\mat{A}^i)_{tt}\leq (n-1)^i$ for any $i\geq 4$, we have that
		\begin{eqnarray*}
			\sum_{i=4}^\infty \abs{(-\frac{1}{\kappa})^i \cdot (\mat{A}^i)_{tt}}
			\leq 
			\sum_{i=4}^\infty \frac{1}{\kappa^i} (\mat{A}^i)_{tt} 
			\leq 
			\sum_{i=4}^\infty \frac{1}{\kappa^i} (n-1)^i
			\leq
			\frac{0.9}{\kappa^3}.
		\end{eqnarray*}
		
		Furthermore, it holds that $(\mat{A}^0)_{tt}=\mat{I}_{tt}=1$; that $\mat{A}_{tt}=0$ since $H$ is a simple graph; and that $(\mat{A}^2)_{tt}=\deg_H(t)=Y$, where the last equation follows from the definition of $Y$. Therefore, 
		\begin{itemize}
			\item If $H$ contains a triangle incident to $t$, then $(\mat{A}^3)_{tt}\geq 2$, and thus
			\begin{eqnarray*}
				(\kappa\cdot \mat{B}^{-1})_{tt} \leq 1+\frac{Y}{\kappa^2}-\frac{2}{\kappa^3} + \frac{0.9}{\kappa^3} = 1+\frac{Y}{\kappa^2}-\frac{1.1}{\kappa^3}
			\end{eqnarray*}
			\item If $H$ has no triangle incident to $t$, then $(\mat{A}^3)_{tt}=0$, and thus
			\begin{eqnarray*}
				(\kappa\cdot \mat{B}^{-1})_{tt} \geq 1+\frac{Y}{\kappa^2} - \frac{0.9}{\kappa^3}
			\end{eqnarray*}
		\end{itemize}
		
		That is, if $H$ contains a triangle incident to $t$, then $(\mat{B}^{-1})_{tt} \leq \frac{1}{\kappa}+\frac{Y}{\kappa^3}-\frac{1.1}{\kappa^4}$. If $H$ does not contain a triangle incident to $t$, then $(\mat{B}^{-1})_{tt} \geq \frac{1}{\kappa}+\frac{Y}{\kappa^3}-\frac{0.9}{\kappa^4}$.
		
		Now let $\Lambda=\mathcal{E}_G(s,t)$. Then by the same argument for proving Claim~\ref{claim:trace}, we have that $\Lambda=(\mat{B}^{-1})_{tt}.$ 
		
		Finally, by the above two claims, if $\vect{u}\mat{M}\vect{v}=1$, then $H$ contains a triangle incident to $t$, and thus $\Lambda=(\mat{B}^{-1})_{tt} \leq \frac{1}{\kappa}+\frac{Y}{\kappa^3}-\frac{1.1}{\kappa^4}$; if $\vect{u}\mat{M}\vect{v}=0$, then $H$ does not contain any triangle incident to $t$, and thus $\Lambda=(\mat{B}^{-1})_{tt} \geq \frac{1}{\kappa}+\frac{Y}{\kappa^3}-\frac{0.9}{\kappa^4}$. The statement of the lemma then follows by the fact that $\lambda$ is a $(1+\frac{1}{\kappa^4})$-approximation of $\Lambda$ and that $\frac{1}{\kappa^4}(\frac{1}{\kappa}+\frac{Y}{\kappa^3}-\frac{0.9}{\kappa^4})\leq \frac{0.1}{\kappa^4}$.
	\end{proof}
\end{proof}

\bibliographystyle{alpha}
\bibliography{literature}

\newpage
\appendix
\begin{center}\huge\bf Appendix \end{center}

\section{Electrical Flows and Effective Resistances}\label{app:energy} Let $G=(V,E,\vect{w})$ be an undirected weighted graph such that $\vect{w}(e)>0$ for any $e\in E$. We fix an arbitrary orientation of edges and treat $G$ as a \emph{resistor network} such that each edge $e\in E$ represents a resistor with \emph{resistance} $\vect{r}(e):=1/\vect{w}(e)$. For any vertex pair $s,t$, the $s-t$ flow is a function $\vect{f}: E \rightarrow \R^+$ satisfying the \emph{conservation condition}, i.e., for any vertex $v\in V\setminus\{s,t\}$, $\sum_{u:(u,v)\in E}\vect{f}(u,v)=\sum_{u:(v,u)\in E} \vect{f}(v,u)$. The \emph{energy of an $s-t$ flow} is defined as $\EE_{G}(\vect{f},s,t):=\sum_{e\in E}\vect{r}(e)\vect{f}(e)^2$. The \emph{$s-t$ electrical flow} $\vect{f}^*$ is defined as the $s-t$ flow that minimizes the energy $\EE_{G}(\vect{f},s,t)$ among all $s-t$ flows $\vect{f}$ with unit flow value, i.e., $\sum_{v\in V}\vect{f}(s,v)=1$. Let $\EE_G(s,t)$ denote the energy of the $s-t$ electrical flow, that is, $\EE_{G}(s,t):=\EE_{G}(\vect{f}^*,s,t)$. An electrical flow $\vect{f}$ naturally corresponds to a \emph{potential} $\phi$ in the sense that we can assign each vertex $u$ a potential $\phi(u)$ such that for any $e = (u,v)$, $\vect{f}(e)=\frac{\phi(u)-\phi(v)}{\vect{r}(e)}$. 

It is well known that the $s-t$ effective resistance $R_G(s,t)$ as defined in Section~\ref{sec:intro} satisfies that $R_G(s,t)=\phi(s)-\phi(t)$, which is the potential difference between $s,t$ when we send one unit of the (unique) $s-t$ electrical flow from $s$ to $t$. Furthermore, it holds that for any $s,t$, the energy of the $s-t$ electrical flow is equivalent to the $s-t$ effective resistance, that is, $\EE_{G}(s,t)=R_G(s,t)$~(see e.g., \cite{doyle84}). In the following, we will mainly focus on how to dynamically maintain (approximation of) effective resistance $R_G(s,t)$.

\section{Missing Proofs from Section~\ref{sec: prelim}} \label{app: prelim}

\begin{proof}[Proof of Lemma~\ref{lem: sepTreeTime}]
For some constant $c \geq 1$, let $S(G)$ be a $\alpha$-balanced separator of size $c\sqrt{n}$, where $\alpha = 2/3$. First, we let $G$ be the root node of $\mathcal{T}(G)$. Let $G_1$ and $G_2$ be the two disjoint components of $G$ obtained after the removal of the vertices in $S$.  We define the children $\child_1(G), \child_2(G)$ of $G$ as follows: $V(\child_i(G)) = V(G_i) \cup S(G)$, $E(\child_i(G)) = E(G_i)$, for $i =1,2$, and whenever an edge connects two vertices in $S(G)$, we arbitrarily append it to either $E(\child_1(G))$ or $E(\child_2(G))$. By construction, property (2) in the definition of $\mathcal{T}(G)$ holds. We continue by repeatedly splitting each child $\child_i(G)$ in the same way as we did for $G$, until there are $O(\sqrt{n})$ components, each of size $O(\sqrt{n})$. The components at this level form the \emph{leaf nodes} of $\mathcal{T}(G)$. Note that the height of $\mathcal{T}(G)$ is bounded by $O(\log n)$ as the size of any child of a node $H$ is at most $2/3$ fraction of the size of $H$.

We define the boundary vertices for each node in $\mathcal{T}(G)$ according to property (3) in the definition of separator trees. To get the bound on the number of boundary vertices per node $H \in \mathcal{T}(G)$, note that the size of $\partial(H)$ is bounded by 
\[
	\left(c \cdot \sum_{i=0}^{O(\log n)}\sqrt{(2/3)^i}\right) \sqrt{n} = O(\sqrt{n}).
\]   


Finally, let $t(n)$ be the maximum time required to construct the separator tree of a $O(\sqrt{n})$-separable graph with $n$ vertices. Then, for some suitably chosen $n_0$, 

\[ t(n) \leq
  \begin{cases}
    s(n) + \max\{t(n_1) + t(n_2)\}  & \quad \text{if } n >  n_0,\\
    0 & \quad \text{if } n \leq n_0,
  \end{cases}
\]
where the maximum is over $n_1, n_2$ such that
\begin{align*}
	& n   \leq n_1 + n_2 \leq n + 2c\sqrt{n}, \quad \text{ and } \quad  \frac{1}{3}n \leq n_i \leq \frac23 n + c\sqrt{n} \quad \text{ for } i = 1,2.
\end{align*}

By a similar analysis as the proof of Theorem 1 of~\cite{EppsteinGIS96}, we can guarantee that $t(n) \leq O(s(n) \log n)$.


\end{proof}

\begin{proof}[Proof of Lemma~\ref{lem: equivPseudo}]
Since $\mat{L}(G)$ is symmetric we can diagonalize it and write \[ \mat{L}(G) = \sum_{i=1}^{n-1} \lambda^{G}_i \vect{u}_i \vect{u}_i^{T}, \]
where $\lambda^{G}_1 \geq \ldots \geq \lambda^{G}_{n-1}$ are the non-zero sorted eigenvalues of $\mat{L}(G)$ and $\vect{u}_1,\ldots,\vect{u}_{n-1}$ are a corresponding set of orthonormal eigenvectors. The \emph{Moore-Penrose Pseudoinverse} of $\mat{L}(G)$ is then defined as 
\[
   \mat{L}(G)^\dagger = \sum_{i=1}^{n-1} \frac{1}{\lambda^{G}_i} \vect{u}_i \vect{u}_i^{T}.
\]

We next show that for every $\vect{x} \in \mathbb{R}^{n}$, $(1-\varepsilon)\vect{x}^{T}\mat{L}(G)\vect{x} \leq \vect{x}^{T}{\mat{L}(H)}\vect{x}$ is equivalent to $\vect{x}^{T}\mat{L}(H)^\dagger\vect{x} \leq \frac{1}{(1-\varepsilon)}\vect{x}^{T}\mat{L}(G)^\dagger\vect{x}$. The other equivalence can be shown in a symmetric way. 

For every $\vect{x} \in \mathbb{R}^{n}$, by definition of $\mat{L}(G)$ and $\mat{L}(H)$ we have \[(1-\varepsilon)\vect{x}^{T}\mat{L}(G)\vect{x}  \leq \vect{x}^{T}{\mat{L}(H)}\vect{x} \Longleftrightarrow (1-\varepsilon)\sum_{i=1}^{n-1} \lambda^{G}_i(\vect{u}_i^{T} \vect{x})^{2} \leq \sum_{i=1}^{n-1} \lambda^{H}_i(\vect{u}_i^{T} \vect{x})^{2}.\]

We next show that 
\begin{equation} \label{eq: equivalPseudo}
	\forall \vect{x} \in \mathbb{R}^{n},~(1-\varepsilon)\sum_{i=1}^{n-1} \lambda^{G}_i(\vect{u}_i^{T} \vect{x})^{2} \leq \sum_{i=1}^{n-1} \lambda^{H}_i(\vect{u}_i^{T} \vect{x})^{2} \Longleftrightarrow (1-\varepsilon)\lambda^{G}_i \leq \lambda^{H}_i,~ \forall i=1,\ldots,n-1.
\end{equation}

Since for every $\mat{x} \in \mathbb{R}^{n}$, $(\vect{u}_i^{T}\vect{x})^{2} \geq 0$, $i=1,\ldots,n-1$, the if-direction of the equivalence in (\ref{eq: equivalPseudo}) follows immediately. For the only-if direction, we proceed by contraposition. To this end, assume that there exists some $i \in \{1,\ldots,n-1\}$ such that $(1-\varepsilon) \lambda_i^{G} > \lambda_i^{H}$. Then there exists a vector $\vect{x} = \vect{u_i} \in \mathbb{R}^{n}$ such that 
\[
   	(1-\varepsilon)\sum_{i=1}^{n-1} \lambda^{G}_i(\vect{u}_i^{T} \vect{x})^{2} = (1-\varepsilon) \lambda_i^{G} > \lambda_i^{H} = \sum_{i=1}^{n-1} \lambda^{H}_i(\vect{u}_i^{T} \vect{x})^{2},
\]
where the first and last inequality follow from the fact that $\vect{u}_i$'s are orthonormal eigenvectors, i.e., $\vect{u}^T_i \vect{u}_i = 1$ and $\vect{u}^T_i \vect{u}_j = 0$, $\forall i \neq j$. This gives a contradiction and thus proves the only-if direction. 
Now, for every $\vect{x} \in \mathbb{R}^{n}$ we have

\begin{align*}
(1-\varepsilon)\vect{x}^{T}\mat{L}(G)\vect{x}  \leq \vect{x}^{T}{\mat{L}(H)}\vect{x} & \Longleftrightarrow (1-\varepsilon)\lambda^{G}_i \leq \lambda^{H}_i,~ \forall i=1,\ldots,n-1 \\
& \Longleftrightarrow \frac{1}{\lambda^{H}_i} \leq \frac{1}{(1-\varepsilon)}\cdot \frac{1}{\lambda^{G}_i},~ \forall i=1,\ldots,n-1 \\
& \Longleftrightarrow \sum_{i=1}^{n-1} \frac{1}{\lambda^{H}_i}(\vect{u}_i^{T} \vect{x})^{2} \leq \frac{1}{(1-\varepsilon)}\sum_{i=1}^{n-1} \frac{1}{\lambda^{G}_i}(\vect{u}_i^{T} \vect{x})^{2} \\
& \Longleftrightarrow \vect{x}^{T}\mat{L}(H)^\dagger\vect{x} \leq \frac{1}{(1-\varepsilon)}\vect{x}^{T}\mat{L}(G)^\dagger\vect{x},
\end{align*}
where the penultimate equivalence can be proven in a similar way to equivalence in (\ref{eq: equivalPseudo}).
\end{proof}

\section{Missing Proofs from Section~\ref{sec:asc_vs}} \label{app: approxSchur}

\begin{proof}[Proof of Lemma~\ref{lemma:exact_schur}]
	Let $k = |K|$. First, note that by Definition~\ref{def: approxSchurComplement} and Lemma~\ref{lem: equivPseudo} we have
	\[
	\forall \vect{x} \in \mathbb{R}^{k},~\frac{1}{(1+\varepsilon)}\vect{x}^{T}\mat{S}(G,K)^\dagger \vect{x} \leq \vect{x}^{T}{\mat{L}(H)}^\dagger \vect{x} \leq \frac{1}{(1-\varepsilon)}\vect{x}^{T}\mat{S}(G,K)^\dagger \vect{x}. 
	\]
	
	Next, let $(s,t) \in K$ be any terminal pair. Consider the demand vector $\vect{\chi}_{s,t} \in \mathbb{R}^{k}$ and extend this vector to $\vect{\chi}_{s,t}' = \begin{bmatrix} \vect{0} &  \vect{\chi}_{s,t} \end{bmatrix}^T \in \mathbb{R}^{n}$. By definition of effective resistance and Lemma \ref{lemm: SchurComp}  we get that 
	\begin{align*}
	R_H(s,t)  & = \vect{\chi}_{s,t}^{T}{\mat{L}(H)}^\dagger \vect{\chi}_{s,t} \\ 
	& \leq \frac{1}{(1-\varepsilon)} \vect{\chi}_{s,t}^{T} \mat{S}(G,K)^\dagger \vect{\chi}_{s,t} \\ 
	& = \frac{1}{(1-\varepsilon)}\vect{\chi}_{s,t}'^{T} \mat{L}(G)^\dagger \vect{\chi}_{s,t}' \\
	& = \frac{1}{(1-\varepsilon)} R_G(s,t).
	\end{align*}
	For the lower-bound on $R_H(s,t)$, using the same reasoning, we get that
	\begin{align*}
	R_H(s,t)  & = \vect{\chi}_{s,t}^{T}{\mat{L}(H)}^\dagger \vect{\chi}_{s,t} \\ 
	& \geq \frac{1}{(1+\varepsilon)} \vect{\chi}_{s,t}^{T} \mat{S}(G,K)^\dagger \vect{\chi}_{s,t} \\ 
	& = \frac{1}{(1+\varepsilon)}\vect{\chi}_{s,t}'^{T} \mat{L}(G)^\dagger \vect{\chi}_{s,t}' \\
	& = \frac{1}{(1+\varepsilon)} R_G(s,t).
	\end{align*}
	
\end{proof}

\section{Missing Proofs from Section~\ref{sec: lowerBound}} \label{app: lowBound}
	\begin{proof}[Proof of Claim~\ref{claim:trace}]
	Let $\mat{L}$ denote the Laplacian matrix of $G$ and let $\vect{v}\in \mathbb{R}^{V_\mat{M}\cup\{t\}}$ denote the vector with entries corresponding to weights between $s$ and $u$ for each $u\in V_\mat{M}\cup\{t\}$, i.e., $\vect{v}_u = \kappa-\deg_H(u)$. 
	
	Now the key observation is that 
	\begin{equation*}
	\mat{L}=\left(\begin{array}{cc}
	\mat{B} & -\vect{v}\\
	-\vect{v}^T & \deg_G(s)
	\end{array}	\right)
	\end{equation*}
	
	For any $\vect{x}\in\mathbb{R}^{V_\mat{M}\cup\{t\}\cup\{s\}}$, let $\widehat{\vect{x}}\in\mathbb{R}^{V_\mat{M}\cup\{t\}}$ be the vector containing the first entries corresponding to vertices in $V_\mat{M}\cup\{t\}$ of $\vect{x}$. Let $\vect{y}$ be the solution of the Laplacian system $\mat{L}\vect{y}=\1_s-\1_t$. Thus, $\vect{y}=\mat{L}^\dag(\1_s-\1_t)$. It also holds that 
	\begin{eqnarray*}
		\mat{B}\cdot \widehat{\vect{y}} -\vect{v}\cdot y_{s} = -\widehat{\vect{1}}_t
	\end{eqnarray*}
	In addition, we know that $\mat{L}\vect{1}=\vect{0}$, and thus $\mat{B}\cdot \widehat{\1}=\vect{v}$. This further implies that, $\widehat{\vect{y}}=\mat{B}^{-1}\cdot \vect{v}\cdot y_s-\mat{B}^{-1}\widehat{\1}_t=y_s\cdot \widehat{\1}-\mat{B}^{-1}\widehat{\1}_t$. Thus, 
	\begin{eqnarray*}
		(\1_s-\1_t)^T \mat{L}^{\dag} (\1_s-\1_t)=(\1_s-\1_t)^T\vect{y}=y_s-\widehat{\1}_t^T\cdot \widehat{\vect{y}}=y_s-\widehat{\1}_t^T\cdot (y_s\cdot \widehat{\1}-\mat{B}^{-1}\widehat{\1}_t) =\widehat{\1}_t^T\mat{B}^{-1}\widehat{\1}_t
	\end{eqnarray*}
	
	Therefore, 
	\begin{eqnarray*}
		\Lambda=\mathcal{E}_G(s,t) = (\1_s-\1_t)^T \mat{L}^{\dag} (\1_s-\1_t) =  \widehat{\1}_t^T\mat{B}^{-1}\widehat{\1}_t =(\mat{B}^{-1})_{tt}
	\end{eqnarray*}
\end{proof}

\end{document}